\documentclass[lettersize,journal]{IEEEtran}
\usepackage{color}
\usepackage{verbatim}
\usepackage{amsfonts}
\usepackage{amssymb}
\usepackage{stfloats}
\usepackage{cite}
\usepackage{graphicx}
\usepackage{subfigure}
\usepackage{amsmath}
\usepackage{array}
\usepackage{epstopdf}
\usepackage{authblk}
\usepackage{graphicx} 
\usepackage{amsthm} 
\usepackage{lipsum}
\usepackage{verbatim} 
\usepackage{authblk}
\usepackage{mathtools}
\usepackage{cuted}
\usepackage[lined,boxed,ruled]{algorithm2e}
\usepackage{algpseudocode}
\usepackage{framed} 
\usepackage{subfigure}
\usepackage{soul}
\usepackage{bm}
\usepackage{setspace}
\usepackage{hyperref}
\hypersetup{hidelinks}

\def\phi{\varphi}

\def\({\left(}
\def\){\right)}

\def\b0{{\mathbf{0}}}

\newtheorem{Lemma}{Lemma}

\newtheorem{Proposition}{Proposition}
\newtheorem{Remark}{Remark}

\newtheorem{Example}{Example}

\definecolor{LatestRevision}{rgb}{0.32, 0.18, 0.5}

\title{Resource Management for Low-latency Cooperative Fine-tuning of Foundation Models at the Network Edge}
\author{
    Hai Wu, Xu Chen, and Kaibin Huang
    \thanks{
        Part of this article has been accepted by the 2024 IEEE International Conference on Communications in China Workshops.
        The authors are with the Department of Electrical and Electronic Engineering, The University of Hong Kong, Hong Kong SAR (Email: \{wuhai, chenxu, huangkb\}@eee.hku.hk). 
        Corresponding author: K. Huang.
    }
}
\makeatletter
\newcommand{\removelatexerror}{\let\@latex@error\@gobble}
\makeatother
\IEEEoverridecommandlockouts

\begin{document}

\maketitle
\vspace{-10pt}
\begin{abstract}

The emergence of large-scale foundation models (FoMo's) that can perform human-like intelligence motivates their deployment at the network edge for devices to access state-of-the-art \emph{artificial intelligence} (AI).
For better user experiences, the pre-trained FoMo's need to be adapted to specialized downstream tasks through fine-tuning techniques.
To transcend a single device's memory and computation limitations, we advocate multi-device cooperation within the \emph{device-edge cooperative fine-tuning} (DEFT) paradigm, where edge devices cooperate to simultaneously optimize different parts of fine-tuning parameters within a FoMo.
The edge server is responsible for coordination and gradient aggregation.
However, the parameter blocks reside at different depths within a FoMo architecture, leading to varied computation latency-and-memory cost due to gradient backpropagation-based calculations. 
The heterogeneous on-device computation and memory capacities and channel conditions necessitate an integrated \emph{communication-and-computation} (C$^2$) allocation of local computation loads and uplink communication resources to achieve \emph{low-latency} (LoLa) DEFT.
To this end, we consider the depth-ware DEFT block allocation problem.
The involved optimal block-device matching is tackled by the proposed low-complexity \emph{Cutting-RecoUNting-CHecking} (CRUNCH) algorithm, which is designed by exploiting the monotone-increasing property between block depth and computation latency-and-memory cost.
Next, the joint bandwidth-and-block allocation (JBBA) makes the problem more sophisticated, i.e., mathematically NP-hard. 
We observe a splittable Lagrangian expression through the transformation and analysis of the original problem, where the variables indicating device involvement are introduced to decouple the block and bandwidth allocation.
Then, the \emph{dual ascent} method is employed to tackle the JBBA problem iteratively. 
Within each iteration, block allocation and bandwidth allocation are optimized concurrently.
The optimal block allocation sub-problem is solved efficiently by applying the Hungarian method facilitated by the proposed CRUNCH algorithm.
On the other hand, the bandwidth allocation sub-problem is solved in closed form, shedding light on favorable allocations to resource-limited devices. 
Through extensive experiments conducted on the GLUE benchmark, our results demonstrate significant latency reduction achievable by LoLa DEFT for fine-tuning a RoBERTa model.

% \begin{comment}
% \begin{IEEEkeywords}
% \end{IEEEkeywords}  
% \end{comment}

\end{abstract}

\section{Introduction}
\label{sec: introduction}

The remarkable advancements in \emph{Artificial Intelligence} (AI), particularly the human-like intelligence as demonstrated by large-scale \emph{foundation models} (FoMo's), have drawn significant investments and research efforts from leading ICT companies such as Microsoft, Alphabet, and Nvidia \cite{Foundation_Model_2021, GPT-3_2020}.
In particular, ChatGPT, a popular generative AI conversation model, attracts over 100 million active users per week \cite{ChatGPT}.
On the other hand, in the era of mobile networks, their \emph{sixth-generation} (6G) is expected to fully integrate AI in various applications, particularly capitalizing on the surge of FoMo's \cite{Toward_Edge_2020, EAI_6G_2022, Edge_AI_2020}.
To enhance the performance of pretrained FoMo's, fine-tuning is widely embraced as a small fraction of model parameters using small datasets, thereby maximizing their utility for the desired downstream tasks \cite{IDE_2021}.
Especially, relevant techniques, collectively called \emph{parameter-efficient fine-tuning} (PEFT), are capable of achieving comparable performance as full-model training by adapting less than 1\% parameters in pretrained FoMo \cite{LoRA, P-tuning_2023}.
Their high efficiencies have led to PEFT being widely adopted in the on-device tuning of FoMo, e.g., Apple Intelligence \cite{Apple_intellignce}.
However, the distributed deployment of fine-tuning in a wireless network to leverage mobile data is still hindered by the high dimensionality of FoMo parameters, the heterogeneous computation-storage capacities of edge devices, and limited radio resources \cite{LLM_meet_6g_2024, Big_AI_model, DEFT_2023},
To address the issue, we propose a novel resource management framework for a multi-device cooperative fine-tuning paradigm that distributes computation loads to devices with awareness of their heterogeneous capabilities for computation, memory, and communication to expedite the adaptation of a FoMo to a specific task.

In wireless mobile networks, there are two approaches to orchestrate the fine-tuning of FoMo's with the cooperation among devices and servers.
One involves uploading local data for \emph{centralized fine-tuning} at the server.
The other approach, called \emph{device-edge fine-tuning} (DEFT) as considered in this work, leverages distributed computational resources and data at devices cooperatively update a FoMo \cite{DEFT_2023, LLM_federated_2024}.
The former circumvents devices' computation bottlenecks by centralizing model fine-tuning at edge servers with powerful AI hardware.
Subsequently, task-oriented FoMo's can be downloaded to devices to enable matching on-device intelligent applications \cite{MBA}.
However, extensive data uploading introduces significant communication delay and raises privacy concerns due to the exposure of raw personal data \cite{DEFT_2023, MEC_survey_2023}.
The server may not be able to guarantee timely and accurate online fine-tuning when requesting devices are many while increasingly powerful on-device processors remain underutilized.
DEFT, to be deployed on an edge-computing platform, can address the drawbacks of centralized fine-tuning by ensuring personal data security and relying on local computation \cite{DEFT_2023}.
It also eases network congestion by avoiding data uploading and bringing computing to the proximity of user data.
Additionally, such proximity allows flexible fine-tuning of FoMo's to adapt to users' changing contexts and applications.
FoMo's cooperatively updated by multiple devices are automatically shared among them without explicit downloading.
Their limitations in computation and communication are also surpassed in the cooperative process.

DEFT is considered a particularization of \emph{mobile edge computing} (MEC), which leverages advanced mobile processors to process the data closest to their sources \cite{MEIC_2020, MEC_survey_2023}.
There exists a rich and mature literature investigating communication-and-computation (C$^{2}$) resources management of MEC in wireless networks, primarily focusing on efficiently offloading computation tasks to edge clouds to enhance communication/computation rates and mobile energy efficiencies \cite{MEC_Energy_5G_2016, MEC_Energy_2017, MEC_Rate_Maximization_2018}.
Subsequently, MEC research was expanded to include AI inference, training, and sensing tasks \cite{Pushing_AI_2023}.
Cooperative transmission of data for split edge inference/sensing \cite{MEC_Cooperative_2020, Comp_and_Comm_Dingzhu, MEC_Partition_2020} and the optimal placement of deep neural networks (DNNs) \cite{MEC_Placement_2020} have been investigated to achieve high computational efficiencies and accuracies for AI tasks over finite-rate wireless links.
A popular class of techniques for training DNN models, called \emph{federated learning} (FL), have been developed to leverage local data and on-device computation while preserving privacy by avoiding direct data uploading, which shares the principles of DEFT \cite{FL_2016, Reliable_FL_2019, Hierarchical_FL, WC_FL_2020}.
However, in each of the FL iterations, every participating device is required to update the entire model and upload the associated stochastic gradients, which is infeasible for devices given the enormous size of a FoMo that ranges from several to hundreds of billion parameters \cite{LLM_federated_2024, DEFT_2023}.
An alternative approach, called \emph{split learning}, where a DNN model is cut into multiple parts for placement on different devices to sequentially compute corresponding stochastic gradients for model updating based on back-propagation (see, e.g., \cite{Split_learning, HiveMind, split_learning_6g}).
However, the transmission of high-dimensional intermediate features for dependent computation among devices results in excessive overhead, especially given the typically large number of iterations.
For example, the feature dimension of a RoBERTa base model using a batch size of 32 is 12.5 million, which is over $500\times$ of the number of LoRA fine-tuning parameters (i.e., 0.024 million) in one associated block \cite{LoRA}.
Meanwhile, deep fading channels of devices are difficult to cope with by scheduling due to the rigid sequential computation order of devices.
Even though researchers have attempted on parallel computing by creating training pipelines, its practicality for mobile networks is questionable due to the uniform requirements on hardware and communication links for devices \cite{HiveMind, PipeDream}. 

The considered multi-device DEFT approach divides and schedules the fine-tuning parameter blocks within a FoMo for edge devices to simultaneously optimize different parts of tunable parameters using local data. 
Compared with FL, devices' computation-and-memory requirements are dramatically lower as they are required to tune a small fraction of model parameters (e.g., 1\%) instead of full model updating.
Moreover, the communication overhead in the context of DEFT is proportional to the number of tunable parameters and, thus, also much less than that for split learning.
Unlike split learning constrained by dependent computing, device scheduling can be employed for DEFT to cope with fading as well as mobile limitations, which is actually the theme of this work.
However, while DEFT has several advantages, its efficient implementation faces a new challenge arising from the fact that the on-device gradient computation for an individual parameter block of a FoMo incurs memory cost and latency that are dependent on the position (i.e., depth) of the block \cite{Memory-efficient, Sublinear_memory, Pytorch_training}.
Specifically, based on the chain rule, the gradients of a specific block are computed by back-propagation from the last block, resulting in the memory-latency cost being monotone-increasing functions of the block depth \cite{Sublinear_memory}.
In the context of FoMo training, the discrepancy among blocks is amplified as each attention block in a transformer is overparameterized, e.g., e.g., 1.8 billion parameters in one transformer block of the GPT-3 model \cite{GPT-3_2020}.
While this fact is largely overlooked in traditional distributed learning targeting relatively small models, it is important to consider the variations of memory-and-latency cost over blocks when designing multi-device DEFT that allocate blocks to devices for local updating.
To be specific, an optimal allocation strategy should match the memory/latency cost of individual blocks with devices' memory/computation capacities and channel states.
This creates a new problem for DEFT termed \emph{block-device matching problem}.

In this work, we make an attempt to solve this problem by optimizing the joint C$^2$ resource management strategy to minimize the fine-tuning latency, yielding the proposed \emph{low-latency DEFT} (LoLa-DEFT) framework.
Such a problem is a mixed integer problem and NP-complete.
We overcome the challenge by designing a low-complexity depth-aware block allocation algorithm named \emph{Cutting-RecoUNting-CHecking} (CRUNCH) that employs a key property of the problem that local computation latency-and-memory cost is a monotone increasing function of block depth.
This allows the search space to be dramatically reduced by shrinking the set of qualified block-device pairs under a latency constraint. 
The optimal policy derived from solving the block-device matching problem is also further developed to facilitate the optimal bandwidth allocation.
The key contributions and findings of this paper are summarized as follows.

    $\bullet$ \textbf{Block Allocation for LoLa DEFT.}
    Given equal bandwidth allocation, the block-device matching problem is transformed into a sequence of feasibility check sub-problems with reduced complexity. 
    Then each sub-problem under a latency constraint from a discrete set involves finding the optimal device-block pairs to make DEFT feasible.
    The sub-problem can be efficiently solved using the proposed CRUNCH algorithm thanks to the exploitation of the mentioned monotone relationship between latency-and-memory cost and block depth.
    The solution nesting a linear search and CRUNCH algorithm is shown to achieve much lower complexity as opposed to the classic Hungarian methods.
    
    $\bullet$ \textbf{Joint Bandwidth-and-Block Allocation for LoLa DEFT}.
    Next, we consider an advanced version of the block-device matching problem incorporating optimal bandwidth allocation.
    The result is a mixed-integer non-linear programming one that is NP-hard.
    By analyzing the problem structure, we proposed an efficient solution approach based on a variation of the Dual Ascent method to compute the optimal policy for joint bandwidth-and-block allocation (JBBA) iteratively. 
    Within each iteration, the Lagrangian optimization is divided into two steps, namely block allocation and bandwidth allocation.
    The former relies on the Hungarian method facilitated by the proposed CRUNCH algorithm to reduce the search space.
    The latter is solved in close form by connecting the device involvement and associated bandwidth. 
    The policy suggests that devices with unfavorable channel conditions or high computation latency/memory cost should be compensated for by favorable bandwidth allocation, which accelerates the DEFT process.
    
    $\bullet$ \textbf{Performance}
    The performance of the proposed block allocation and JBBA strategies for LoLa-DEFT is evaluated through extensive experiments on fine-tuning a RoBERTa model on the GLUE benchmark \cite{GLUE}. They demonstrate significant improvement over the baseline without considering block depth for model training, e.g., a 40\% latency reduction in fine-tuning a RoBERTa base model.

The rest of this paper is organized as follows. 
In Section~\ref{sec: models and operations}, DEFT models and operations are elaborated. 
Then, the problems of block-device matching and JBBA for LoLa-DEFT are formulated in Section~\ref{sec: problem formulation}. 
The low-complexity solution approaches are investigated in Section~\ref{sec: assigned bandwidth} and Section~\ref{sec: joint bandwidth allocation}, respectively.
Experimental results are provided in Section~\ref{sec: experiments}, followed by concluding remarks in Section~\ref{sec: conclusions}.

\section{Models and Operations}
\label{sec: models and operations}

\begin{figure*}
    \centering
    \includegraphics[width=0.75\textwidth]{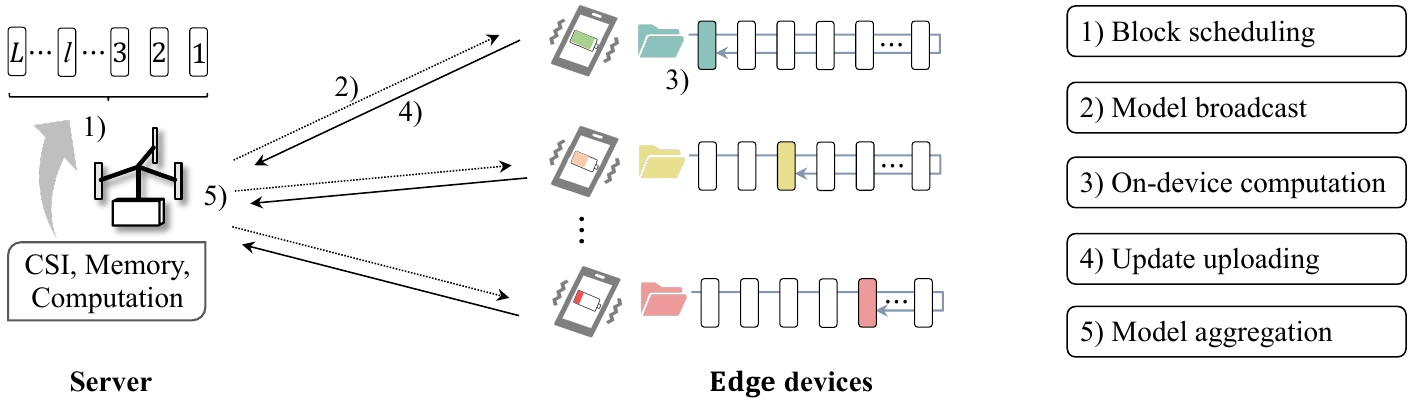}
    \vspace{-4mm}
    \caption{The multi-device cooperative fine-tuning system and operations within one communication round.}
    \label{fig: system model}
    \vspace{-4mm}
\end{figure*}

Consider a single-cell system where $K$ edge devices cooperatively fine-tune a FoMo for a given downstream task, as illustrated in Fig.~\ref{fig: system model}.
The devices utilize on-device computation power and localized data to compute the gradient of tunable parameters in the FoMo using the \emph{stochastic gradient descent} (SGD) method.
The indices of these devices are denoted as $\mathcal{K} = \{1, 2, \cdots, K\}$. 
Fine-tuning parameter division and depth-aware parameter block allocation for multi-device cooperative computation are adopted to accelerate the fine-tuning process.
Since the DEFT requires a repeated aggregation of on-device computed gradients to update the model parameters, the edge server needs to perform depth-aware block allocation by jointly considering local computation and gradient upload at the beginning of each communication round. 
Hence, within each communication round, the activated devices, need to perform the following three steps: 
a) download one initialized/aggregated FoMo;
b) compute the gradients of a scheduled FoMo block in the model iteratively;
c) upload the locally computed FoMo gradients.  
After the server successfully receives all parameter gradients, the latest updated FoMo is offloaded to all scheduled devices to start a new round. 
The workflow of DEFT is summarized in Algorithm~\ref{algorithm: DEFT workflow}, with detailed models and operations described in the following sub-sections.

\begin{algorithm}[t] \caption{Workflow of the Multi-device DEFT.}
\label{algorithm: DEFT workflow}
    \textbf{Initialization:} Initialize fine-tuning parameter blocks $\{\phi_l^{(0)}\}_{l=1, \cdots, L}$. \\
    \textbf{for} communication round $n=1, \cdots, N$ \textbf{do} \\
    \quad \textbf{Server:} \\
    \quad \quad Get devices' memory $\{\hat{b}_k^{(n)}\}$, CSI $\{H_{k}^{(n)}\}$, \\
    \quad \quad ~~~~computation coefficient $\{f_k\}$.\\
    \quad \quad Perform depth-aware block allocation (and \\
    \quad \quad ~~~~bandwidth allocation).\\
    \quad \quad Broadcast latest model $\{\phi_{l}^{(n)}\}$ and scheduling \\
    \quad \quad ~~~~instructions $\{\alpha_{k,l}^{(n)}\}$ (and $\{B_k\}$).\\
    \quad \textbf{Devices:} \\
    \quad \quad Check scheduling instructions \\
    \quad \quad \textbf{for} activated devices \textbf{do}: \\
    \quad \quad \quad Fetch the latest model. \\
    \quad \quad \quad \textbf{for} local iteration $m=1, \cdots, M$ \textbf{do}: \\
    \quad \quad \quad \quad Perform gradient computation to optimize \eqref{eq: local fine-tuning} \\
    \quad \quad \quad \quad ~~~~based on localized dataset.  \\
    \quad \quad \quad \quad Terminate gradient propagation after the \\
    \quad \quad \quad \quad ~~~~scheduled block's gradient is obtained. \\
    \quad \quad \quad \textbf{end for} \\
    \quad \quad \quad Upload the computed gradient $\{g^{(n)}(\phi_l)\}$. \\
    \quad \quad \textbf{end for} \\
    \quad \textbf{Server:} \\
    \quad \quad Aggregate the gradients of parameter blocks and \\
    \quad \quad ~~~~update the FoMo based on \eqref{eq: model aggregation}. \\
    \textbf{end for} \\
    \textbf{Output:} the cooperatively fine-tuned $\{\phi_l^{(N)}\}_{l=1, \cdots, L}$. \\
\end{algorithm}

\vspace{-5mm}
\subsection{Fine-tuning Model}
\label{subsec: fine-tuning model}

In the considered DEFT, the tunable parameters are evenly distributed over model blocks (e.g., Transformer-based FoMos) as commons adopted in the literature and practice \cite{Scale_2023, LoRA, Adapter_2019, bitfit_2022}.  
Suppose the given pre-trained foundation model, denoted by $P_{\Phi} (y|x)$, is parameterized by $\Phi$.
$P_{\Phi} (y|x)$ is a generic multi-task learner such as GPT.
We follow the fine-tuning downstream task setup in \cite{LoRA} to adapt $P_{\Phi} (y|x)$ to conditional text generation task, e.g., natural language understanding. 
In centralized fine-tuning, the desired task is characterized by a training dataset consisting of context-target token sequence pairs, which is denoted by $\mathcal{Z} = \{(x_i, y_i)\}_{i=1, \cdots, |\mathcal{Z}|}$.
For example, $x_i$ is the content of an article and $y_i$ is its summary in the text summarization task.
A set of tunable fine-tuning parameters, denoted by $\phi$, associated with $L$ blocks within one transformer, i.e., $\phi = \{\phi_1, \cdots, \phi_L\}$, is optimized through SGD-based optimization by repeatedly feeding the training data into $P_{\Phi} (y|x)$, to maximize the conditional language modeling objective.
$\phi$ can be exactly the original foundation model parameter $\Phi$ in the full-model tuning case, or $\phi$ is the newly appended fine-tuning parameters in PEFT settings.
Overall, fine-tuning task is hence optimizing $\phi$ through maximizing the conditional language modeling objective, i.e., the probability of generating the desired token given the conditioning context and previously generated words,
\begin{equation}
    \max_{\phi} \sum_{(x,y) \in \mathcal{Z}} \sum_{j=1}^{|y|} \log \big ( P_{\Phi, \phi} (y_j | x,y_{<j}) \big)
    \label{eq: fine-tuning}.
\end{equation}
In the considered multi-device cooperative fine-tuning, an arbitrary device, say device $k$, holds a localized training dataset $\mathcal{Z}^{(k)} = \{(x_i^{(k)}, y_i^{(k)})\}_{i=1, \cdots, |\mathcal{Z}^{(k)}|}$.
The localized data is assumed to be independent and identically distributed (i.i.d.). 
Hence, the asynchronous update of different parameter blocks can be performed using localized data in a distributed manner \cite{ASGD}.  
Thereby, the fine-tuning parameter is obtained through the optimization in \eqref{eq: fine-tuning} by replacing the centralized dataset $\mathcal{Z}$ with $ \mathcal{\hat{Z}} = \cup_{k=1}^{K} \mathcal{Z}^{(k)}$.

\vspace{-5mm}
\subsection{On-device Computation Model}
\label{subsec: on-device computation model}

\begin{figure*}
    \centering
    \includegraphics[width=0.73\textwidth]{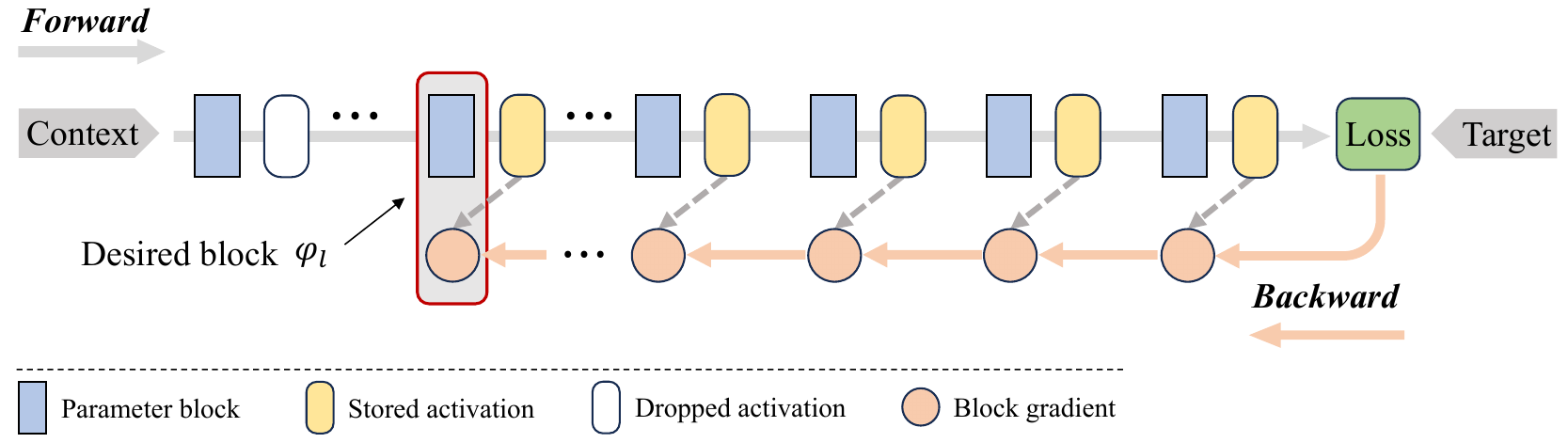}
    \vspace{-3mm}
    \caption{On-device computation of parameter gradient in device $k$ based on the gradient backpropagation approach. 
    All the activations, i.e., forward intermediate results and block gradients, are recorded in memory for chain rule-based gradient calculation. 
    For multi-device DEFT, one activated device only needs to record the intermediate activations associated with the desired block and its gradient for server update. 
    Once the desired parameter block is calculated, the backpropagation is terminated and the devices can fetch a new local iteration.
    }
    \label{fig: on-device computation}
    \vspace{-5mm}
\end{figure*}

\subsubsection{On-device Gradient Backpropagation-based Fine-tuning}
In the considered FoMo division for DEFT, each device is responsible for computing the gradient of one parameter block in the entire $\phi$.
Consider an arbitrary device, say device $k$, is scheduled to compute the gradient of a specific fine-tuning parameter block $\phi_l$, denoted as $g(\phi_l)$, toward the following on-device optimization:
\begin{equation}
    \max_{\phi_l} \sum_{(x,y) \in \mathcal{Z}^{(k)}} \sum_{j=1}^{|y|} \log \big ( P_{\Phi, \phi_{l}} (y_j | x,y_{<j}) \big )
    \label{eq: local fine-tuning}.
\end{equation}

Consider the commonly adopted backpropagation approach for gradient computation as shown in Fig.~\ref{fig: on-device computation}.
The training data is fed and propagated through the entire foundation model to calculate the loss function, i.e., the objectives in \eqref{eq: local fine-tuning}. 
There is a need to store all the weights, intermediate activations and gradients along the propagation of the whole model in SGD-based fine-tuning of foundation models.
Then, the gradient of the loss function with respect to (w.r.t.) the fine-tuning parameter, denoted by $g(\phi_{l})$, is obtained by the repeated employment of chain rule from the gradient of the last layers towards the former layers.
Specifically, the gradient of the former layer's parameters is determined by its intermediate activation saved in the forward path and the later layer's gradient.

\subsubsection{Memory-aware On-device Cooperative Computation}
To transcend the on-device memory limitation, the fine-tuning parameter division allows the heterogeneous devices to cooperatively compute the gradient of fine-tuning parameters residing in different depths according to the corresponding memory budgets. 
% Moreover, devices can adopt gradient checkpointing to record only a few intermediate activations after the scheduled tuning block instead of storing all of them \cite{Sublinear_memory}, making the on-device fine-tuning memory efficient at the sacrifice of the on-the-fly forwarding of intermediate activation for gradient calculation.
Once the device completes the gradient calculation of the scheduled model block, the backward path is terminated.
Given the $L$ block-repetitive transformer architecture, the server needs to schedule $L$ block among $K$ devices ($L \leq K$).
Specifically, the allocation of a parameter block with depth index $l$, i.e., $\phi_l$, to device $k$ in round $n$ is denoted by $\alpha_{k,l}^{(n)} = 1$, otherwise $\alpha_{k,l}^{(n)}=0$.
The on-device memory (in bits) required to compute the gradient of fine-tuning parameters with block index $l$ is denoted by $b(l)$, where $b(\cdot)$ is a monotone-increasing function of $l$.
For example, $b(l)$ is a linear function over $l$ if none of the memory-saving techniques are adopted \cite{Sublinear_memory}.
% , or it could be a linear function over $\sqrt{l}$ if the checkpointing is adopted \cite{Sublinear_memory}.

\subsubsection{On-device Computation Latency}
To achieve a balanced fine-tuning of all parameters, the on-device gradient computation within one round is performed with the same number of iterations, $M$, and the same batch size of training samples is adopted.
The latency bottleneck of on-device computation is the backpropagation process, which triples the latency of a forward path, see \cite{Pytorch_training}. 
The sequential computation of a substantial number of parameters within each model block necessitates the integration of diverse backpropagation lengths, i.e., $l \in \{L, \cdots, 1\}$, for multi-device scheduling.
Therefore, the computation latency of one device is dominated by its computation capability and the depth index of the assigned parameter block. 
Given a device with normalized computation capability, e.g., CPU frequency is 1GHz, the run-time duration (in seconds) of computing the gradient of block $l$ is $d(l) = a + c(l)$, where $a$ represents the latency of data feeding and the entire model forward, $c(l)$ is the latency of backpropagation ended at block $l$.
$c(l)$ is a monotone-increasing function of $l$ and generally is a linear one.
Consider the heterogeneous computation capabilities of edge devices, each with a relative computation factor of the normalized one, denoted by $f_k, \forall k$.
Then, the on-device computation latency of scheduling device $k$ for fine-tuning model block $l$ in round $n$ can be written as
\begin{equation}
    T_{k, l, \sf comp}^{(n)} = \alpha_{k, l}^{(n)} \frac{M d(l)}{f_k} \triangleq \alpha_{k,l}^{(n)} J_{k,l}.
\end{equation}

\subsection{Communication Model}

\subsubsection{Model Downloading}
At the beginning of each communication round, the server first broadcasts the latest foundation model to all involved devices. 
Since the server is associated with the base station, the broadcasting rate could be relatively high. 
Therefore, the downlink latency can be regarded as the same among devices and is considered to be negligible.

\subsubsection{Gradient Uploading}
Once the device finishes computations within one round, the accumulated gradient of $M$ local iterations is uploaded to the server for central aggregation.
The devices are connected to the edge server via wireless channels, where the \emph{orthogonal frequency-division multiple access} (OFDMA) is adopted for allocating a broadband spectrum to involved devices.
The channels are assumed to be frequency non-selective and there is a sufficiently large number of sub-carriers, making the bandwidth allocation can be approximated as being continuous.
The total bandwidth is represented by $B$.
Consider an arbitrary device, say device $k$.
Its transmit power is represented by $p_k$, which is known by the server. 
The channel power in the $n$-th communication round is denoted by $H_{k}^{(n)}$, which remains a constant within one round and can be perfectly estimated by the server. 
Let $0 \leq B_k^{(n)} \leq B$ denote the bandwidth allocated to device $k$. 
The upload communication rate, $R_k^{(n)}$, can be written as
\begin{equation}
    R_k^{(r)} = B_k^{(n)} \log_2(1 + \frac{p_k H_{k}^{(n)}}{N_0}) \triangleq B_k^{(n)} r_k^{(n)} ,
\end{equation}
where $N_0$ is the white Gaussian noise power, $r_k^{(n)}$ is defined to be the uplink spectrum efficiency.
For the blockwise distributed fine-tuning parameters, its size is assumed to be $S$ (in bits).
Then, the upload latency of scheduling device $k$ targeting the fine-tuning of model block $l$ in round $n$ can be written as
\begin{equation}
    T_{k, l, \sf comm}^{(n)} = \frac{\alpha_{k, l}^{n} S}{B_k^{(n)} r_k^{(n)}}
\end{equation}

\subsection{Global FoMo Updating}

After the server successfully receives the gradients of all the model blocks, i.e., the accumulated gradient of all parameter blocks, in round $n$, denoted by $\{g^{(n)}(\phi_l)\}_{l=1, \cdots, L}$, the server performs centralized model aggregation.
Given the gradient descent is applied, the updated model block $\phi_l$ after round $n$, represented by $\phi_l^{(n)}$ can be written as
\begin{equation}
    \phi_l^{(n)} = \phi_l^{(n-1)} - \eta^{(n)} \mathcal{G} \big(g^{(n)}(\phi_l) \big),
    \label{eq: model aggregation}
\end{equation}
where $\eta^{(n)}$ is the predefined learning rate, $\mathcal{G}(\cdot)$ represents the variation of gradients depending on the selection of different optimizers.
For example, the Adam optimizer relies on the first and second-order moments for adaptive element-wise updates of fine-tuning parameters.

\vspace{-2mm}
\section{Problem Formulation}
\label{sec: problem formulation}

As the core component of LoLa DEFT, the objective of parameter block allocation is to minimize the total fine-tuning latency, considering the devices' heterogeneous computation capability, available memory space, and channel status.  
In this section, the associated block-device matching problem is formulated as two optimization problems for the relatively simple case with given bandwidth allocation and the JBBA case.
Given a fixed number of communication rounds $N$, the latency minimization of the entire fine-tuning process is to minimize the latency of every communication round.
Therefore, the problem can be simplified as a one-round latency minimization problem equivalently. 
For ease of notation, the superscripts indicate the communication round number is omitted in the following.

Consider an arbitrary device, say device $k$. 
The associated latency, $T_k$, consists of two components, i.e., the on-device computation part and the uplink communication part, which can be written as
\begin{equation}
    T_k = \sum_{l=1}^{L} (T_{k, l, \sf comp} + T_{k, l, \sf comm}) = \sum_{l=1}^{L} \alpha_{k, l} \big( J_{k,l} + \frac{S}{B_k r_k} \big).
\end{equation}
For the considered multi-device cooperative system, the overall latency is exactly determined by the device involved with the maximum summed latency, i.e., 
\begin{equation}
    T = \max_{k \in \mathcal{K}} \Big \{\sum_{l=1}^{L} \alpha_{k, l} \big( J_{k,l} + \frac{S}{B_k r_k} \big)\Big\}.
\end{equation}
On top of the target of latency minimization, block allocation requires satisfying a few constraints as stated below.
At the beginning of each communication round, the server is required to schedule $L$ distinct devices from $K$ involved devices to compute the fine-tuning parameter block residing in depth from $l=L$ to $l=1$, which introduces the following spreading model block constraints,
\begin{equation} \mathrm{(C1) ~~~~~~}
    \sum_{k=1}^{K} \alpha_{k, l} = 1, \forall l. \notag
\end{equation}
The constraint ensures that one model block must be allocated to exactly one edge device for fine-tuning.
Moreover, a device is assigned to finish the task of computing the gradient of at most one model block if it is involved, i.e., 
\begin{equation} \mathrm{(C2) ~~~~~~}
    \sum_{l=1}^{L} \alpha_{k, l} \leq 1, \forall k. \notag
\end{equation}
Besides, the model block assignment for involved devices should not exceed the device's available memory, i.e., avoiding hitting the memory wall. 
Considering the memory requirements for different on-device computations, the memory budget of device $k$ is denoted as $\hat{b}_k$, the corresponding constraints are given as
\begin{equation} ~~~~~\mathrm{(C3) ~~~~~}
    \sum_{l=1}^{L} \alpha_{k, l} b(l) \leq \hat{b}_k, \forall k. \notag
\end{equation}

Consider the scenario with equal bandwidth allocation ($B_{k} = B/L$). 
Under the above constraints and the target of one round latency minimization, the block allocation problem for multi-device cooperative fine-tuning is formulated as follows:
\begin{equation*}\mathrm{(P1) \quad}
    \begin{aligned}
        \min \limits_{\{\alpha_{k,l}\}} \qquad 
        & \max_{k \in \mathcal{K}} \quad \sum_{l=1}^{L} \alpha_{k,l} T_{k,l} \\
        \mathrm{s.t.~~~~~} 
        & T_{k,l} = J_{k,l} + \frac{S L}{B r_k},  \forall k, l, \\ 
        & \mathrm{(C1)~\&~(C2)~\&~(C3)}, \\
        & \alpha_{k, l} \in \{0, 1\},  \forall k, l.
    \end{aligned}
\end{equation*}
Next, consider the JBBA case where $\{B_k\}$ is joint optimized with block allocation.
The diversified bandwidth assignment can further align the latency of all involved devices to shorten the overall duration. Specifically, the associated scheduling problem is formulated as follows:
\begin{equation*} \quad \mathrm{(P2) \quad}
    \begin{aligned}
        \min \limits_{\{\alpha_{k,l}\}, \{B_k\}} \quad 
        & \max_{k \in \mathcal{K}} \quad \sum_{l=1}^{L} \alpha_{k,l} T_{k,l} \\
        \mathrm{s.t.~~~~~~} 
        & T_{k,l} = J_{k,l} + \frac{S}{B_k r_k},  \forall k, l, \\ 
        & \sum_{k=1}^{K} B_k \leq B, \\
        & \mathrm{(C1)~\&~(C2)~\&~(C3),} \\
        & \alpha_{k, l} \in \{0, 1\}, \forall k, l. \\
    \end{aligned}
\end{equation*}

\section{Depth-Aware Block Allocation for LoLa-DEFT}
\label{sec: assigned bandwidth}

In this section, we design LoLa-DEFT by solving the block-device pairing problem in $\mathrm{(P1)}$. 
The optimal block allocation is obtained by first transforming the original min-max problem into an equivalent latency-threshold minimization problem and then addressing the later problem with low complexity via leveraging a linear search approach and a proposed depth-aware feasibility checking method.

% By exploiting the depth-latency correlation, we have proposed a low-complexity optimal scheduling algorithm.

% The optimal block allocation is obtained by first transforming the original min-max problem into an equivalent latency-threshold minimization problem and then addressing the later problem via leveraging a linear search approach and a proposed depth-aware feasibility checking method.
\vspace{-5mm}
\subsection{Optimal Block Allocation via Feasible Set Minimization}
For tractability, the min-max problem $\mathrm{(P1)}$ is transformed into a minimization problem by introducing a variable of latency threshold, $T_{\sf th}$, as
\begin{equation*}\mathrm{(P3) \qquad}
    \begin{aligned}
        \min \limits_{\{\alpha_{k,l}\},T_{\sf th}} \quad 
        & T_{\sf th} \\
        \mathrm{s.t.~~~~~} 
        & \sum_{l=1}^{L} \alpha_{k,l} T_{k,l}\leq T_{\sf th}, \forall k,\\
        & \mathrm{Constraints\ in\ (P1)}.\\
    \end{aligned}
\end{equation*}
The equivalence between $\mathrm{(P1)}$ and $\mathrm{(P3)}$ is guaranteed by the joint optimization of $\{\alpha_{k,l}\}$ and $T_{\sf th}$. 
It is observed from $\mathrm{(P3)}$ that the feasible set of $T_{\sf th}$ can be restricted into the discrete set $\{T_{k,l}\}$. 
Hence, the optimal latency threshold towards $\mathrm{(P3)}$ can be obtained by searching over the C$^2$ latency $T_{k,l}$ of all device-block pairs. 
For each searched latency threshold $T_{\sf th}$, the feasibility of associated problem $\mathrm{(P3)}$ requires validation through the finding of the maximum matching on the left qualified device-block pairs.
Hence, the workflow of finding the optimal block allocation can be summarized into the following four steps: 
\begin{itemize}
    \item \textbf{Step 1:} Rearrange $\{T_{k,l}\}$ in a descending order to generate a set $\{\tilde{T}_j\}$ with $j \in \{1,\cdots,KL\}$ and $\tilde{T}_j\geq \tilde{T}_{j+1}$, $\forall j$; set a variable $t = KL+1$ and initialize a set recording solution of $\{\alpha_{k,l}\}$ as $\mathcal{P}=\emptyset$.
    \item \textbf{Step 2:} Update $t = t-1$ and set $T_{\sf th} = \tilde{T}_t$.
    \item \textbf{Step 3:} Validate if there is a feasible solution of $\{\alpha_{k,l}\}$ under the current latency constraint:
    \begin{itemize}
        \item If yes, update the current solution into $\mathcal{P}$ and return to Step 2; 
        \item Otherwise, go to Step 4. 
    \end{itemize}
    \item \textbf{Step 4:}  Output $\mathcal{P}$ as the optimal solution.
\end{itemize}
The operation of such linear search, due to the discrete search $\{T_{k,l}\}$, results in polynomial complexity of $\mathcal{O}(KL)$ in the worst case. However, the feasibility validation in each iteration causes the main complexity. 
To be specific, the devices and FoMo blocks are treated as two disjoint sets in a bipartite graph, formulating the feasibility validation as a matching problem. 
The feasibility validation is then to find a pairwise matching between devices and blocks under the constraints of ensuring all blocks are matched with different devices. 
The resulting matching offers a possible block assignment, whose feasibility is determined by checking the constraints of $\sum_{l=1}^{L} \alpha_{k,l} T_{k,l}\leq T_{\sf th}$ under the current requirement $T_{\sf th} = \tilde{T}_t$. 
The conventional Hungarian algorithm can be leveraged to perform the said bipartite matching. 
However, even using its most efficient variant, the Hungarian algorithm requires a computation complexity of $\mathcal{O}(\max\{K,L\}^3)$, together with the operation of linear search, resulting in the overall complexity of $\mathcal{O}(KL\max\{K,L\}^3)$. Such complexity is not acceptable when a large number of devices or blocks are involved.

% \subsubsection{Depth-Aware Feasibility Checking for Low-Complexity Scheduling}  
\vspace{-5mm}
\subsection{Depth-aware Low-complexity Matching Validation}

On top of resorting to solving the bipartite matching, we propose a low-complexity feasibility-validation method that exploits the correlation between block depth and latency. 
To be specific, the memory and computation latency required by a block monotonically increase as its depth increases. 
This endows the matching bipartite graph of the feasibility-validation problem with a special geometric structure that allows for low-complexity matching. 

Before introducing the low-complexity feasibility-validation criterion, we first define the auxiliary notations.
Typically, considering a given latency constraint $T_{\sf th}$, a $K$-by-$L$ matrix, denoted by $\mathbf{Q}_{T_{\sf th}}$, is initialized to characterize both memory constraints and C$^2$ latency of each fine-tuning block at each device. 
The elements of $\mathbf{Q}_{T_{\sf th}}$ are computed as 
\begin{equation}\label{eq:intial}
    q_{k,l} = \begin{cases}
       T_{k,l},  & b(l)\leq \hat{b}_k \mathrm{\ and\ } T_{k,l}\leq T_{\sf th}, \\
       0,  & \mathrm{otherwise},
    \end{cases}
\end{equation}
where $q_{k,l} = 0$ reveals that the $l$-th block can not be allocated to device-$k$ due to its limited memory or unacceptable latency cost. 
It is noticed that for an arbitrary device, both its resulting latency, $T_{k,l}$, and memory costs, $b(l)$, will monotonically increase as the block depth $l$ increases. 
Hence, the monotone increasing property makes elements in a single row of $\mathbf{Q}_{T_{\sf th}}$ follow the property of
\begin{equation}\label{eq:ordering}
    q_{k,l}\geq q_{k,l^{\prime}}>0,\quad \mathrm{if\ } l\geq l^{\prime}.
\end{equation} 
Then, the number of non-zero elements in each row of $\mathbf{Q}_{T_{\sf th}}$ can be counted as
\begin{equation}\label{eq:counting}
    s_k = |\mathcal{S}_k|,\  \mathcal{S}_k \overset{\triangle}{=}\{q_{k,l}\ |\ q_{k,l}\neq 0,\forall l\}.
\end{equation}
Let $\pi(j)$ denote the $j$-th smallest value in $\{s_k\}$ and $\kappa_{\pi(j)}$ indicate the device index corresponding to $\pi(j)$. 
That is, $s_{\kappa_{\pi(j)}}$ is the $j$-th smallest value in $\{s_k\}$. 
Then, the positive correlation between computation costs and block depth is exploited to realize low-complexity feasible validation as follows.

\begin{Proposition}[Feasibility Checking]\label{Prop:condition}
    \emph{The problem $\mathrm{(P3)}$ has at least one feasible solution with the resulting latency being smaller than or equal to $T_{\sf th}$ if and only if 
    \begin{equation}\label{eq:condition}
        \pi(K-L+l)\geq l, 1\leq l\leq L.
    \end{equation}}
\end{Proposition}
\begin{proof}
    We complete the proof by verifying the adequacy and necessity of~\eqref{eq:condition} as follows.
    % \begin{itemize}
    
        $\bullet$ Meeting~\eqref{eq:condition} provides a feasible solution for $\mathrm{(P3)}$ with $T_{\sf th}$: Given the latency constraint $T_{\sf th}$, the possible user-block pairs are identified by the non-elements in $\mathbf{Q}_{T_{\sf th}}$. Based on the property of~\eqref{eq:ordering}, the user $\kappa_{\pi(j)}$ can be associated with the blocks of $1\leq l\leq j$, subject to the latency and memory constraints. With the condition~\eqref{eq:condition} being satisfied, there is $K-L+l\geq l$, $\forall l$ due to $K\geq L$. Therefore, one can allocate the block $l$ to the user $\kappa_{\pi(K-L+l)}$.
        
        $\bullet$ Feasible $T_{\sf th}$ for $\mathrm{(P3)}$ enforces~\eqref{eq:condition} to hold: Consider the block $l$ allocated to an arbitrary user, say $k$, there shall be $\alpha_{k,l} = 1$ and $q_{k,l}\geq q_{k,l-1}\geq\cdots\geq q_{k,1}>0$ according to~\eqref{eq:ordering}. Then, based on the definition of~\eqref{eq:counting}, there is $s_k = l$. Next, $T_{\sf th}$ is feasible for $\mathrm{(P3)}$ means that all blocks will be allocated to users one by one, with the user responsible for block $l$ having at least $l$ non-zero elements in its corresponding row in $\mathbf{Q}_{T_{\sf th}}$. Due to $K\geq L$, this leads to $\pi(K-L+l)\geq l, 1\leq l\leq L$.
    % \end{itemize}
\end{proof}
Consequently, according to Proposition~\ref{Prop:condition}, given the condition~\eqref{eq:condition} hold, a qualified device-block pairing policy is given by: 
\begin{equation}\label{eq:assignment}
    \alpha_{\kappa_{\pi(K-L+l)},l}=l,\forall l.
\end{equation}

\subsection{Cutting-Recounting-Checking (CRUNCH) Algorithm}

Based on the feasible set minimization mechanism and the low-complexity feasibility check described above, the transformed problem $\mathrm{(P3)}$ can be solved efficiently via the algorithm elaborated below.
The algorithm is established on the matrix $\mathbf{Q}_{T_{\sf th}}$ and consists of three iterative operations: \textbf{C}utting, \textbf{R}eco\textbf{UN}ting, and \textbf{CH}ecking (CRUNCH). 
\begin{itemize}
    \item[1)] \textbf{Cutting:} updates the matrix $\mathbf{Q}_{T_{\sf th}}$ by removing its largest element (set as zero), which removes a feasible value of $T_{\sf th}$ from its feasible set consists of $\{T_{k,l}\}$.
    \item[2)] \textbf{Recounting:} updates the number of non-zero elements in each row of $\mathbf{Q}_{T_{\sf th}}$, i.e., $s_k$, according to~\eqref{eq:counting}.
    \item[3)] \textbf{Checking:} validates the feasibility of $T_{\sf th}$ through the obtained $\{s_k\}$ and Proposition~\ref{Prop:condition}. 
\end{itemize}
The above three steps are repeated until the feasibility is violated, and the final $T_{\sf th}$ gives the optimal latency. 
The above Cutting-RecoUNting-CHecking (CRUNCH) operations are concluded into Algorithm~\ref{algorithm: crunch}, which is illustrated through an example below.

\begin{algorithm}[t] \caption{CRUNCH Algorithm for Optimal Block Assignment.}
\label{algorithm: crunch}
    \textbf{Initialization:} $t=1$, $T_{\sf th} = \inf$, construct $\mathbf{Q}_{T_{\sf th}}$ based on~\eqref{eq:intial}, and compute $s_k^{(0)}$ using $\mathbf{Q}_{T_{\sf th}}$. \\
    \textbf{Repeat:} \\
    \quad \textbf{Cutting:} find $(k^{\star},l^{\star}) = \arg\max_{k,l}\  q_{k,l}$, record \\
    \quad ~~~~$T_{\mathsf{temp}} = q_{k^{\star},l^{\star}}$, and update $\mathbf{Q}_{T_{\sf th}}$ with $q_{k^{\star},l^{\star}}=0$.\\
    \quad \textbf{Recounting:} update $s_{k^{\star}}^{(t)}=s_{k^{\star}}^{(t-1)}-1$ \\
    \quad ~~~~and $\{\pi(j)\}_{1\leq j\leq K}$.\\
    \quad \textbf{Checking:} check the feasibility of $T_{\sf th}$ based on \\
    \quad ~~~~Proposition~\ref{Prop:condition} and $\{\pi(j)\}_{1\leq j\leq K}$:\\
    \quad \quad \textbf{Yes:} update $t=t+1$, and return to the Cutting \\
    \quad \quad ~~~~step;\\
    \quad \quad \textbf{No:} break iteration, output $T_{\sf th} = T_{\mathsf{temp}}$ and \\
    \quad \quad ~~~~$\{\pi(j)\}_{1\leq j\leq K}$.\\
    \textbf{Output:} assignment $\alpha_{\kappa_{\pi(K-L+l)},l}=l,\forall l$. \\
\end{algorithm}

\begin{figure}[t]
    \centering
    \includegraphics[width=0.48\textwidth]{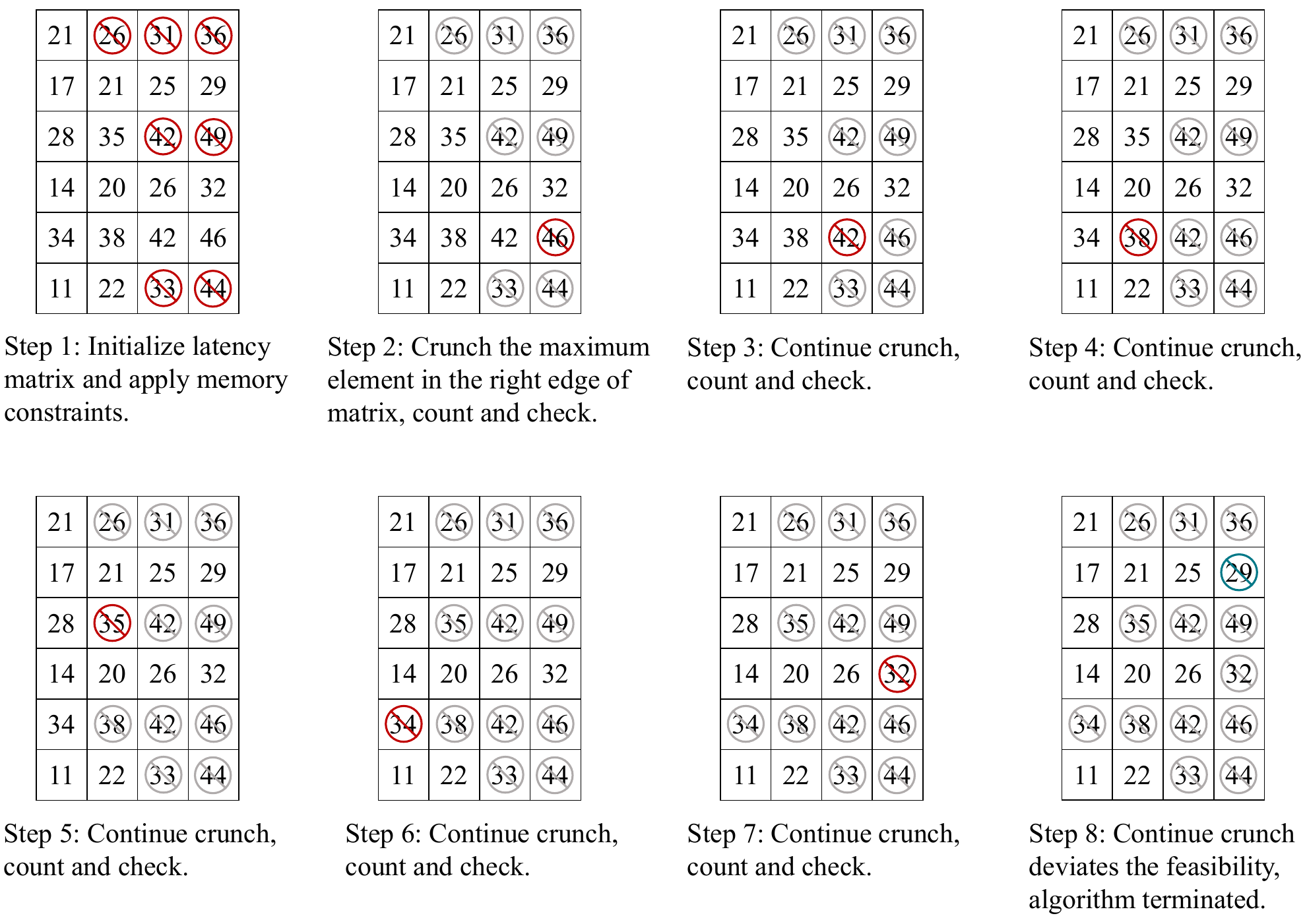}
    \vspace{-2mm}
    \caption{An example of the proposed CRUNCH algorithm.}
    \vspace{-5mm}
    \label{fig: crunch}
\end{figure}

\begin{Example}[Example for Executing CRUNCH]
    \emph{We present an example of block assignment with $K=6$ and $L=4$ in Fig.~\ref{fig: crunch}. During the initialization step, we form the primary matrix $\mathbf{Q}_{T_{\sf th}}$ by removing elements corresponding to user-block pairs that violate memory constraints. Each row of this matrix ensures that the elements are in ascending order from left to right due to the depth property. In the next step, we remove the maximum non-zero element on the right edge of $\mathbf{Q}_{T_{\sf th}}$, which is $46$. As a result, $\{s_k\}$ becomes $\{1, 2, 2, 3, 4, 4\}$, satisfying the feasibility condition stated in Proposition~\ref{Prop:condition}. Consequently, we update $T_{\sf th}$ to $46$ and continue to remove the elements $42$ in Step 3, $38$ in Step 4, $35$ in Step 5, $34$ in Step 6, and $32$ in Step 7. All of these elements pass the feasibility check. Finally, in Step 8, when we remove the $(2,4)$-th element in $\mathbf{Q}_{T_{\sf th}}$, the condition of~\eqref{eq:condition} is violated, leading to the termination of the iteration. Hence, the latency is $29$, and the optimal device-block assignment is $\{(1,1),(6,2),(4,3),(2,4)\}$.}  

\end{Example}

% \subsection{Complexity Analysis}\label{subsec: complexity}
In Algorithm~\ref{algorithm: crunch}, given the initial $\mathbf{Q}_{T_{\sf th}}$, the main complexity in each iteration comes from finding the maximum element in $\{q_{k,l}\}$. Due to the depth-latency correlation, only a single value in each row of $\mathbf{Q}_{T_{\sf th}}$ will be involved in computation, leading to the complexity of feasibility checking being $\mathcal{O}(K)$. 
Then, the overall complexity is $\mathcal{O}(K^2L)$, which is dramatically reduced compared with the conventional method with a complexity of $\mathcal{O}(KL\max\{K,L\}^3)$.

\section{Joint Bandwidth-and-Block Allocation for LoLa-DEFT}
\label{sec: joint bandwidth allocation}

In this section, we further develop LoLa-DEFT to address the practical issue of optimal bandwidth allocation as it helps to accelerate the DEFT process. 
Specifically, we propose a solution approach for the JBBA problem formulated in $\mathrm{(P2)}$. 
It is found that problem $\mathrm{(P2)}$ is a mixed-integer non-linear programming (MINLP), which is an NP-hard problem.
We first introduce the auxiliary variables denoting the device involvement to decouple the block allocation indicators and controllable bandwidth in the original problem.  
By analyzing the Lagrangian of the transformed problem, the \emph{dual ascent method} is designed to obtain the optimal JBBA policy iteratively.
In each iteration, a variation of the designed CRUNCH algorithm is applied to provide sparse qualified device-block pairs, based on which the optimal block allocation is obtained by employing the Hungarian method to find the \emph{minimun-weight perfect matchings}. 
The optimal device involvement and bandwidth allocation are deduced in close-form solutions.

\subsection{Problem Transformation and Analysis}

To facilitate the decomposition of joint bandwidth allocation and block allocation, we first introduce binary decision variables $\beta_k$ to denote the involvement of device $k$ in the considered communication round, where $\beta_k = 1$ if device $k$ is scheduled to compute the gradient of an arbitrary parameter block and $\beta_k = 0$ otherwise.
The constraint $\mathrm{(C2)}$ that one device can be scheduled to compute the gradient of at most one parameter block, which is translated into the constraint of $\beta_k = \sum_{l=1}^{L} \alpha_{k,l}, \forall k $.
Then, based on the device involvement indicator $\beta_k$, the depth-aware block allocation and bandwidth allocation are decoupled in the latency representations.
Hence, the multi-device cooperative fine-tuning latency minimization problem can be rewritten as
\begin{align}
    \min_{\substack{T, \{B_k\}, \{\beta_k\}, \\ \{\alpha_{k,l}\} \in \mathcal{S}_b}} \quad 
    & T \notag \\
    \mathrm{(P4) \quad ~~} \mathrm{s.t.~~~~~~~} 
    & \sum_{l=1}^{L} \alpha_{k,l} J_{k,l} + \beta_k \frac{S}{B_k r_k} \leq T, \forall k, \tag{C4.1} \label{C1}\\
    & \sum_{k=1}^{K} B_k \leq B, \tag{C4.2} \label{C2}\\
    & \beta_{k} = \sum_{l=1}^{L} \alpha_{k,l}, \forall k, \tag{C4.3}  \label{C3}
\end{align}
where the binary depth-aware block allocation indicators $\{\alpha_{k,l}\}$ are constrainted in the feasible set $\mathcal{S}_b$ composed by $\mathrm{(C1)}$, $\mathrm{(C2)} $ and $\mathrm{(C3)}$ listed in problem $\mathrm{(P2)}$.
There are $L$ distinct blocks that should be computed in an arbitrary communication round, which implicitly indicates the constraint on $\{\beta_k\}$, i.e., $\sum_{k=1}^K \beta_k = L$.

Then, the Lagrangian of problem $\mathrm{(P4)}$ is expressed as following
\begin{equation}
    \begin{aligned}
    & \mathcal{L} \big(T, \{B_k\}, \{\beta_k\}, \{\alpha_{k,l}\}, \{\lambda_k\}, \mu, \{\sigma_k\} \big) \\ 
    & \quad = T + \sum_{k=1}^K \lambda_k \big(\sum_{l=1}^L \alpha_{k,l}J_{k,l} + \beta_k \frac{S}{B_k r_k} - T \big) \\
    & \qquad + \mu (\sum_{k=1}^{K} B_k - B) + \sum_{k=1}^{K} \sigma_k (\beta_k - \sum_{l=1}^{L} \alpha_{k,l}),
    \end{aligned}
    \label{eq: Lagrangian}
\end{equation}
where $\{\lambda_k\} \geq 0, \mu \geq 0$ are dual variables associated with the constraints \eqref{C1} and \eqref{C2} respectively, and $\{\sigma_k\}$ are dual variables with constraints \eqref{C3}.
By categorizing the primal values, the Lagrangian \eqref{eq: Lagrangian} can be rewritten as 
\begin{equation}
    \begin{aligned}
    & \mathcal{L} \big(T, \{B_k\}, \{\beta_k\}, \{\alpha_{k,l}\} \big) \\ 
    & \quad = \underbrace{(1 - \sum_{k=1}^K \lambda_k) T}_{\mathcal{L}_T} + \underbrace{\sum_{k=1}^K \sum_{l=1}^L \alpha_{k,l} (\lambda_k J_{k,l} - \sigma_k) }_{\mathcal{L}_{\alpha_{k,l}}} \\
    & ~~~~~~+ \underbrace{\sum_{k=1}^K (\lambda_k \frac{S}{B_k r_k} + \sigma_k) \beta_k + \mu (\sum_{k=1}^K B_k - B)}_{\mathcal{L}_{B_k, \beta_k}}.
    \end{aligned}
    \label{eq: separated Lagrangian}
\end{equation}
It is observed that the block allocation indicators $\{\alpha_{k,l}\}$ are splittable with the device involvement indicators $\{\beta_k\}$ and associated bandwidth $\{B_k\}$, hence can be solved independently, making the \emph{dual decomposition} \cite{ADMM_Boyd} an efficient guideline to tackle problem $(\mathrm{P4})$. 

\begin{algorithm}[t] \caption{Dual Ascent for JBBA}
\label{algorithm: IC2 DEFT}
    \textbf{Initialization:} $\{\lambda_k\} \geq 0, \mu \geq 0, \{\sigma_k\}, T, \{B_k\}, \{\beta_k\}, \{\alpha_{k, l}\} $ \\
    \textbf{Repeat until convergence:} \\
    \quad Update the latency $T$. \\
    \quad Apply CRUNCH algorithm by replacing the element \\
    \quad ~~~~in $\mathbf{Q}$ with $(\lambda_k J_{k,l} - \sigma_k), \forall k, l.$\\
    \quad Apply Hungarian method on the crunched sparse $\mathbf{Q}$ \\
    \quad ~~~~to find optimal assignment $\{\alpha_{k,l}\}$.\\
    \quad Obtain the optimal activated devices and bandwidth \\
    \quad ~~~~allocation using Lemma \ref{Lemma: beta and bandwidth}.\\
    \quad Update the dual variables using \eqref{eq: update dual}.\\
    \textbf{Output:} optimal $T, \{B_k\}, \{\beta_k\}, \{\alpha_{k, l}\} $. \\
\end{algorithm}

\vspace{-5mm}
\subsection{Dual Ascent for JBBA}
\label{subsec: dual ascent}

The strategy of dual ascent is to optimize primal variables and dual variables alternatively.
Specifically, the following two optimization steps are conducted iteratively to obtain the optimal JBBA solutions.

$\bullet$ \textbf{Primal optimization:} Optimizing primal variables $T, \{B_k\}, \{\beta_k\}, \{\alpha_{k,l}\}$ to minimize Lagrangian \eqref{eq: Lagrangian} given fixed dual variables $\{\lambda_k^i\}, \mu^i, \{\beta_k^i\}$ in iteration $i$. The associated primal problem is
    \begin{equation*}\mathrm{(P5)~}
    \begin{aligned}
        \min_{\substack{T, \{B_k\}, \{\beta_k\}, \\ \{\alpha_{k,l}\} \in \mathcal{S}_b}}
        & \mathcal{L}\big(T, \{B_k\}, \{\beta_k\}, \{\alpha_{k,l}\} \big| \{\lambda_k^i\}, \mu^i, \{\sigma_k^i\} \big) \\
        \mathrm{s.t.~~~~} 
        & \beta_{k} = \sum_{l=1}^{L} \alpha_{k,l}, \forall k.
    \end{aligned}
    \end{equation*}

$\bullet$ \textbf{Dual updating:} Apply dual ascent on dual variables to maximize the Lagrangian \eqref{eq: Lagrangian} given the latest obtained primal variables. 
    It is noted that a simple gradient ascend can be applied to dual variables because of their differentiability. 
    Denoting the solution to problem $\mathrm{(P5)}$ in the i-th iteration as $T^i, \{B_k^i\}, \{\beta_k^i\}$ and $\{\alpha_{k,l}^i\}$, the updating of dual variables can be
    \begin{equation}
    \begin{aligned}
        \lambda_k^{i+1} & = \lambda_k^i + \epsilon^i ( \sum_{l=1}^{L}{\alpha_{k,l}^{i}} J_{k,l} + {\beta_k^i} \frac{S}{{B_k^i}} - {T^i}) \\
        \mu^{i+1} & = \mu^{i} + \epsilon^i (\sum_{k=1}^{K} B_k^i - B) \\
        \sigma_k^{i+1} & = \sigma_k^{i} + \epsilon^i (\beta_k^i - \sum_{l=1}^{L} \alpha_{k,l}^i)
    \end{aligned}
    \label{eq: update dual}
    \end{equation}

Hence, the challenge towards the JBBA problem $\mathrm{(P4)}$ lies in tackling problem $\mathrm{(P5)}$ efficiently, where the primal minimization step is decomposed into the minimization of three corresponding separable partial Lagrangians, $\mathcal{L}_T$, $\mathcal{L}_{\alpha_{k,l}}$ and $\mathcal{L}_{\beta_k, B_k}$ concerning the primal variables listed in the subscript respectively.
Obviously, the first term $\mathcal{L}_T$ is a linear function over the latency variable $T$, which induces the achievable minimized round latency $T$, i.e., $\{0, T_{\mathrm{max}}\}$ depend by the sign of the coefficient.
% \begin{equation}
%     T = \begin{cases}
%         0,  & \text{if} \quad (1 - \sum_{k=1}^K \lambda_k) > 0; \\
%         T_{\mathrm{max}}, & \text{if} \quad (1 - \sum_{k=1}^K \lambda_k) \leq 0;
%     \end{cases}
%     \label{eq: L_t}
% \end{equation}
$T_{\mathrm{max}}$ can be an arbitrary prefixed upper bound of the entire latency, e.g., the longest device latency obtained in the last iteration. 
The monotone property of $\mathcal{L}_T$ necessitate the coefficient $(1 - \sum_{k=1}^K \lambda_k)$ tend to 0 over optimization iterations.
The solutions for minimizing the other two partial Lagrangians are detailed in the following subsections.  
The dual ascent-based algorithm solving problem $(\mathrm{P4})$ is summarized in Algorithm~\ref{algorithm: IC2 DEFT}.

\subsection{Block Allocation for Primal Minimization}

The minimization of partial Lagrangian, $\mathcal{L}_{\alpha_{k,l}}$, as introduced in Section~\ref{subsec: dual ascent}, corresponds to the optimal block allocation for JBBA.
It is observed that the problem here is a variant of the previously explored block allocation for a given bandwidth case, i.e., problem $\mathrm{(P1)}$.
The feasibility checking in Proposition~\ref{Prop:condition} is not altered since the co-existence of constraint $\{\alpha_{k,l}\} \in \mathcal{S}_b$. 
The minor difference between these two problems lies in the objective, where problem $\mathrm{(P1)}$ emphasizes the maximum per-device latency while problem $\mathrm{(P5)}$ targets the minimization over summations.
Both aim to find a perfect pairing of parameter blocks within all devices. 

Thereby, the developed CRUNCH Algorithm~\ref{algorithm: crunch} can be employed accordingly.
Specifically, the elements in the constructed matrix $\mathbf{Q}$ are no longer every device's latency but replaced by $(\lambda_k J_{k,l} - \sigma_k), \forall k, l,$ correspondingly, 
After the continuous crunching steps, the remaining elements resided in the crunched matrix $\mathbf{Q}$ before violating the feasibility checking, specify the device candidates for every parameter block. 
Given the crunched matrix $\mathbf{Q}$, the problem $\mathrm{(P1)}$ is solved optimally.
An arbitrary selection of a specific device candidate for every block would not inflect the most stringent latency requirement.
However, the minimization of partial Lagrangian $\mathcal{L}_{\alpha_{k,l}}$ necessitates considering different device selection strategies after obtaining the crunched $\mathbf{Q}$, which can affect the final sum-up result since the concurrence of devices qualified for one block.

The problem then becomes a classical finding of \emph{Minimum Weight Perfect Matching in Bipartite Graphs} \cite{Perfect_matching} within the crunched matrix $\mathbf{Q}$.
The bipartite is formed by the \emph{qualified device set} and the block set, where each connecting edge weighs the value associated in $\mathbf{Q}$ (edges associated with non-existing elements, i.e., unqualified pairs denoted by 0 in $\mathbf{Q}$, are weighted by $+\infty$). 
Then, it can be solved through the existing algorithms, e.g., the Hungarian method, solving the minimum-cost assignment problem \cite{Hungarian_1955}.
It is noted that the Hungarian method can also be directly applied to solve the original problem. 
However, the CRUNCH operations provide a fairly sparse cost matrix for optimal assignment finding.
Up to half of the elements in $\mathbf{Q}$ are exempted from computation in the conduction of the Hungarian method, which accelerates the direct assignment over the intact cost matrix $\mathbf{Q}$.

\subsection{Bandwidth Allocation for Primal Minimization}

We aim to minimize the last partial Lagrangian, $\mathcal{L}_{B_k, \beta_k}$, as introduced in Section~\ref{subsec: dual ascent}, by considering device involvement and bandwidth allocation jointly. The optimal solution is derived below.

\begin{Lemma}\label{Lemma: beta and bandwidth}
    \emph{The $L$ selected device are the ones with $L$ least values in the set of $ \Big \{2 \sqrt{\frac{\lambda_k S \mu}{r_k}} + \sigma_k \Big \}$, the corresponding optimal bandwidth allocation is given by
    \begin{equation}
        {B_k}^{\star} = \sqrt{\frac{{\beta_k}^{\star} \lambda_k S}{r_k \mu}}, \quad \forall k \in \mathcal{K}.
    \end{equation}
    }
\end{Lemma}
\begin{proof}
    The partial derivative of $\mathcal{L}_{B_k, \beta_k}$ with respect to $B_k$ is given by 
    \begin{equation*}
        - \frac{\beta_k \lambda_k S}{{B_k}^2 r_k} + \mu = 0.
    \end{equation*}
    The optimal bandwidth allocation is given by setting the derivative as zero. Then, $\{B_k\}$ is substituted back to $\mathcal{L}_{B_k, \beta_k}$. By utilizing the property of binary variables that $\sqrt{\beta_k} = \beta_k$, the resultant expression of $\mathcal{L}_{B_k, \beta_k}$ regarding $\{\beta_k\}$ is
    \begin{equation}
        \mathcal{L}_{B_k, \beta_k} = \Big (\sum_{k=1}^K \big(2 \sqrt{\frac{\lambda_k S \mu}{r_k}} + \sigma_k \big ) \beta_k  \Big) - \mu B.
        \label{eq: optimal beta}
    \end{equation}
    Recall the constraint on $\{\beta_k\}$ that there should be exactly $L$ devices involved, i.e., $\sum_{k=1}^K \beta_k = L$. Hence, the minimization of \eqref{eq: optimal beta} requires setting $\beta_k = 1$ for the devices with the $L$ least coefficients. 
    This completes the proof.
\end{proof}

\begin{Remark}[Bandwidth Compensation for LoLa-DEFT]
\label{remark: bandwidth allocation} 
\emph{For depth-aware block allocation, the round latency is determined by the worst scheduled device-block pair, which is generally less computationally powerful. 
In JBBA, the involved device would hold a greater C$^{2}$ latency-associated dual variable, i.e., $\lambda_k$, if there is a larger latency gap between it and the high-end one, as depicted in \eqref{eq: update dual}.
The optimal bandwidth allocation for JBBA, as elaborated in the previous Lemma, would compensate the involved and less C$^{2}$-conditioned devices with higher bandwidth consumption to align all involved devices' latency, hence further minimizing the overall latency of multi-device DEFT. }
\end{Remark}

\begin{figure*}[t]
    \centering
    \subfigure[]{\includegraphics[width=0.35\textwidth]{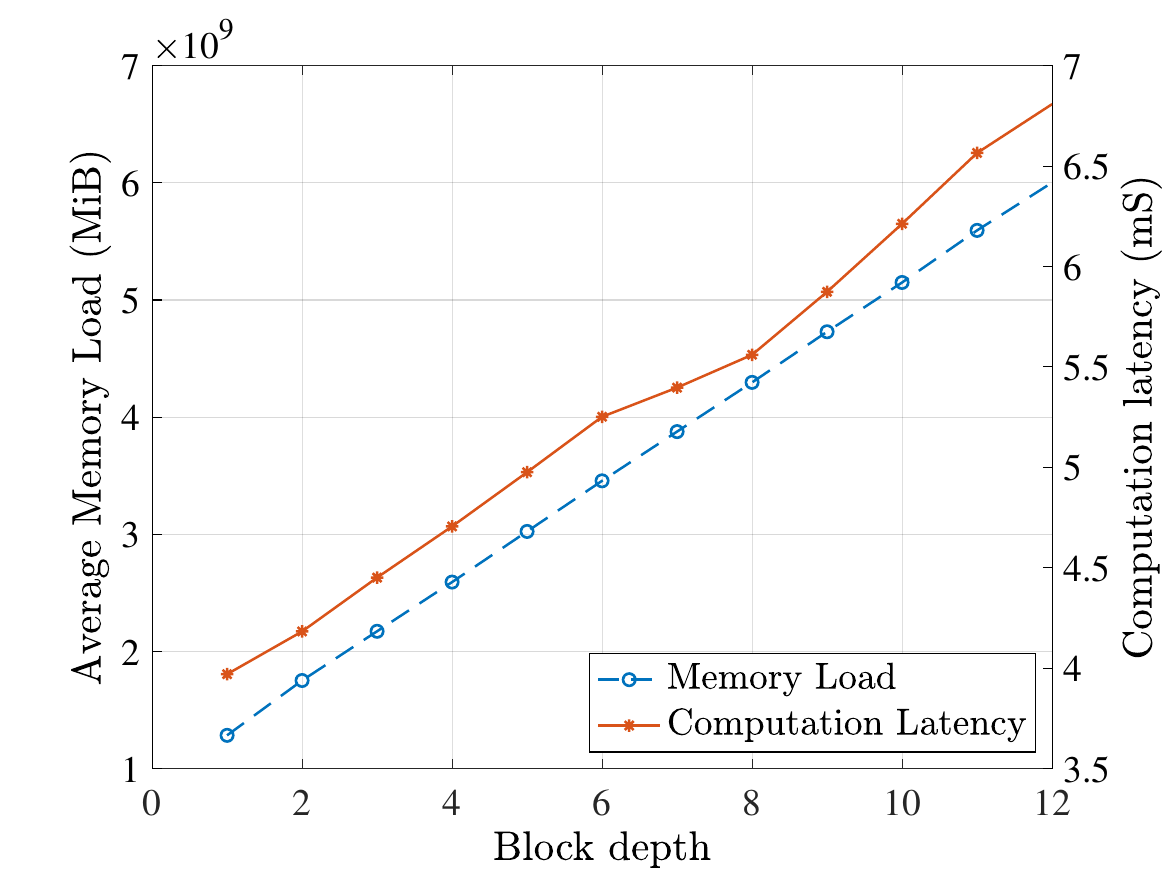}\label{fig: memory latency depth}}
    \hspace{2.5cm}
    \subfigure[]{\includegraphics[width=0.35\textwidth]{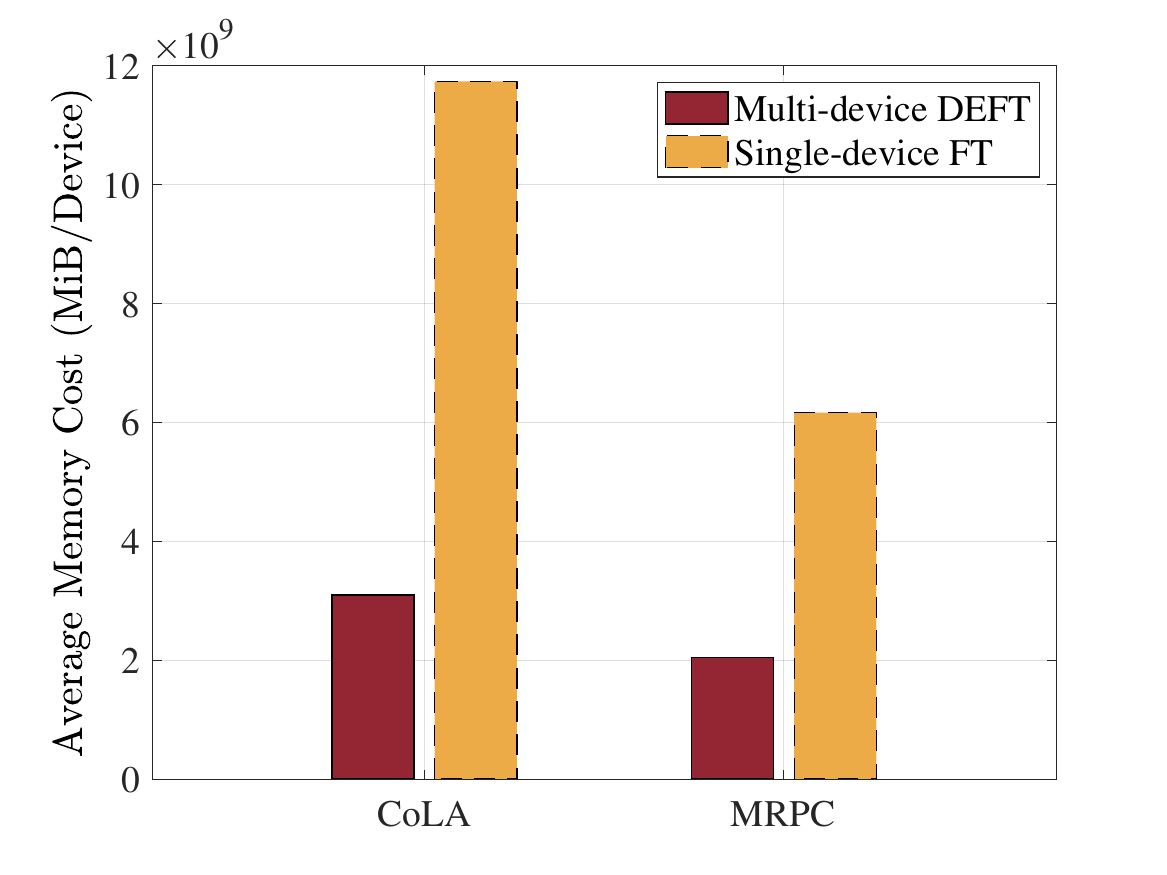}\label{fig: deft memory}}
    \vspace{-2mm}
    \caption{(a) The average on-device memory load and computation latency of fine-tuning the parameter block in different depths of a RoBERTa base model toward the CoLA task. (b) The average memory cost of a device fine-tuning a RoBERTa base model considering single-device computing and multi-device cooperative computing scenarios.
    }
    \vspace{-2mm}
\label{fig: why deft}
\end{figure*}

\vspace{-5mm}
\section{Experimental Results}
\label{sec: experiments}

\subsection{Experiment Setup}

    \subsubsection{Fine-tuning of FoMo's} 
    To verify the effectiveness of the proposed multi-device cooperative fine-tuning paradigm, we conduct experiments based on the low-rank adaptation (LoRA) of the RoBERTa base model.
    The model consists of 12 transformer encoder layers with a hidden size of 768, which are partitioned into 12 blocks accordingly.
    The fine-tuning parameters are two low-dimensional matrices with a rank of 8 for the query and value projection matrix within every block.
    Hence, there are 24576 tunable parameters within one block, each characterized by 32 bits. 
    The GLUE benchmark is the fine-tuning task to evaluate the natural language understanding capabilities over subtasks listed and detailed in \cite{GLUE}.
    We follow the training scripts provided by \cite{LoRA} and accommodate them to our proposed multi-device cooperative fine-tuning paradigm.
    Half-precision training is adopted to save memory costs. 
    The Adam optimizer is employed, and all hyper-parameters are kept the same with centralized tuning for fair comparison.
    All fine-tuning tasks are performed on an A100 80GB server.
    \subsubsection{On-device computation settings} 
    The Snapdragon 8 Gen 3 mobile chip anchors the standard local computation capability.
    The number of floating point operations calculates the latency of local computations over different gradient backpropagation lengths.
    All devices' relative computation capability coefficients are uniformly distributed, i.e., $U(0.5, 1)$.
    All devices' available memory in each round is uniformly distributed in the $[1, 6]$ GB range.
    For depth-aware scheduling, all required memory for fine-tuning every parameter block, ${b}(l), \forall l$, is measured in advance based on the memory consumption of the training server. 
    \subsubsection{Communication settings}
    There are assumed $20$ devices engaged in the entire multi-device DEFT process.
    The channel gains of the devices are assumed to follow i.i.d. Rayleigh fading with a path loss of $10^{-3}$ overall communication rounds.
    The involved device adopts a constant transmit signal-to-noise ratio (SNR) of 10 dB at each round. 
    The bandwidth used for uplink transmission is 100 MHz. 
    For equal bandwidth assignment cases (BA for DEFT), the 12 involved devices equally split the bandwidth budget. 
    For the bandwidth allocation cases (JBBA for DEFT), the upload bandwidth of every device is obtained through the solution of the JBBA problem.
    \subsubsection{Benchmark schemes}
    We present two benchmarking schemes. The computation-ware DEFT activates the devices with the best computation capabilities in each DEFT round and allocates the demanding FoMo blocks according to their computation power.
    The communication-aware DEFT activates the devices with the best communication channel conditions in each round and assigns the blocks accordingly.
    It is noted that the solely computation-aware or communication-aware scheduling of devices for multi-device DEFT could result in hitting the memory wall, making the multi-device cooperative fine-tuning infeasible. 

\vspace{-3mm}
\subsection{Memory-and-computation Efficiency of Multi-device DEFT}

\begin{figure*}[t]
    \centering
    \subfigure[CoLA]{\includegraphics[width=0.35\textwidth]{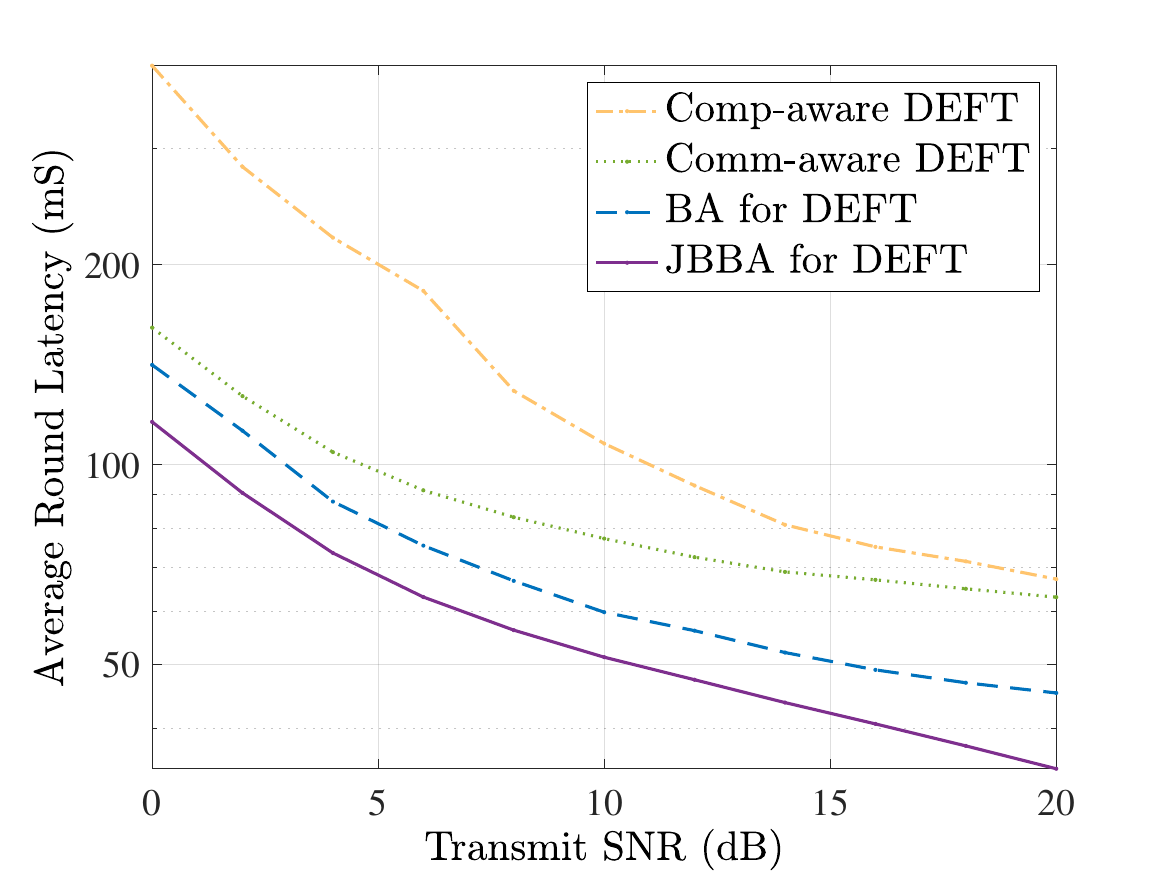}}\label{fig: round_t_snr_cola}
    \hspace{2.5cm}
    \subfigure[MRPC]{\includegraphics[width=0.35\textwidth]{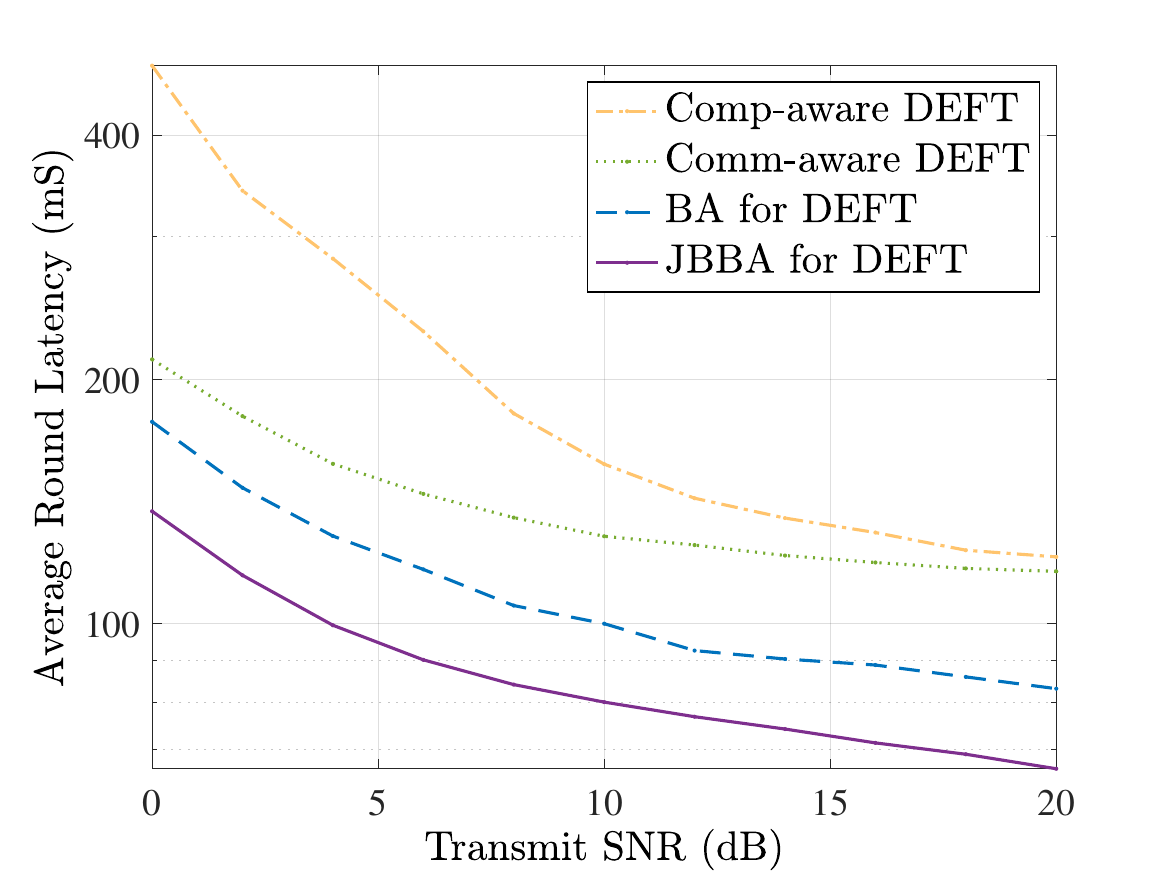}}\label{fig: round_t_snr_mrpc}
    \vspace{-2mm}
    \caption{The averaged round latency of multi-device fine-tuning of RoBERTa w.r.t. the different transmit SNR on the task of (a) CoLA and (b) MRPC, respectively.}
    \label{fig: round_t_snr}
    \vspace{-2mm}
\end{figure*}

\begin{figure*}[t]
    \centering
    \vspace{-5mm}
    \subfigure[CoLA]{\includegraphics[width=0.35\textwidth]{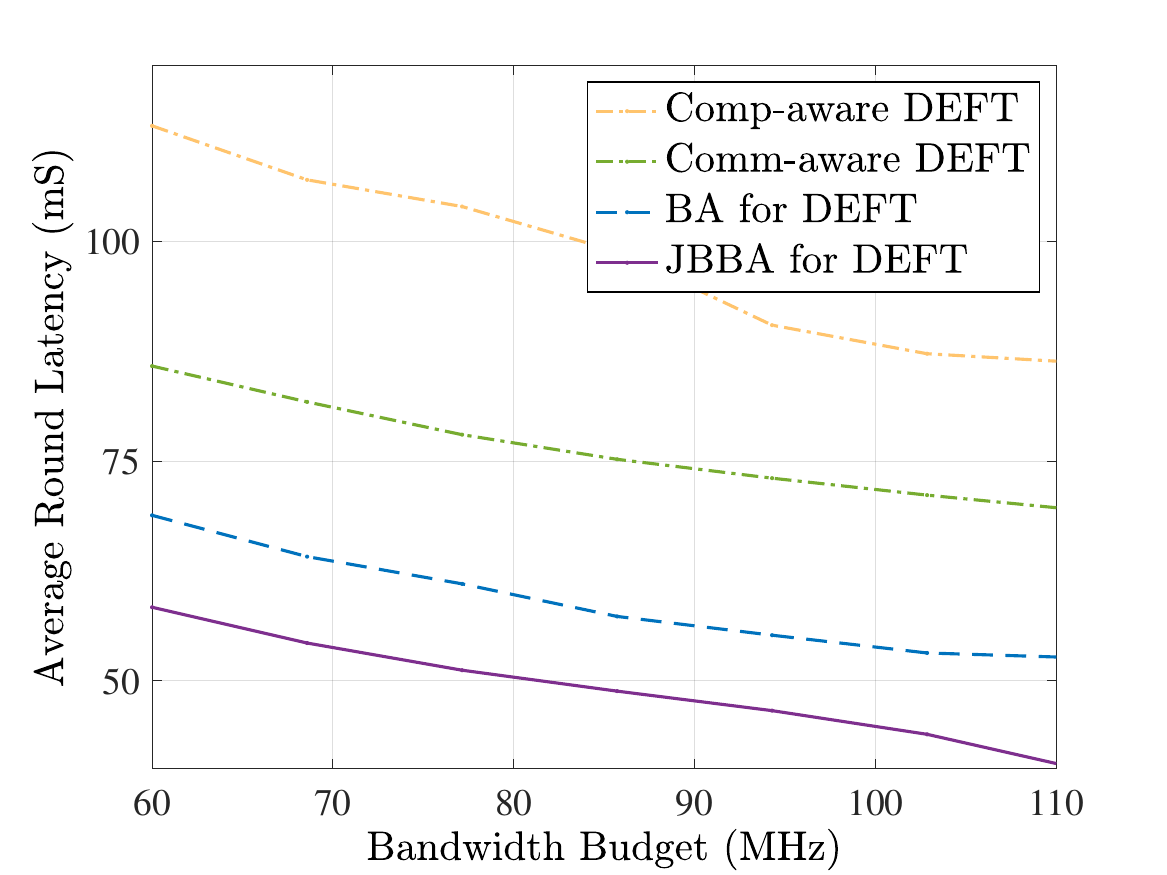}}\label{fig: round_t_B_cola}
    \hspace{2.5cm}
    \subfigure[MRPC]{\includegraphics[width=0.35\textwidth]{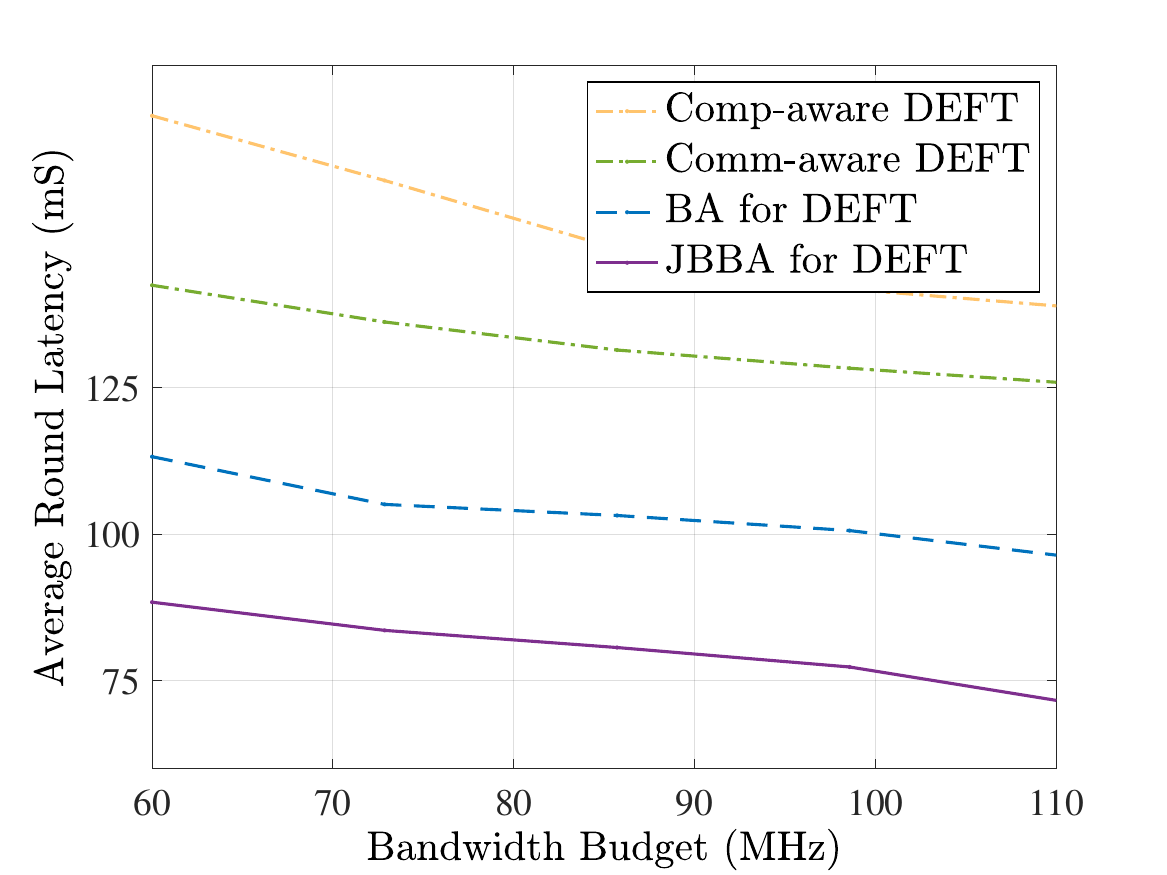}\label{fig: round_t_B_mrpc}}
    \vspace{-2mm}
    \caption{The averaged round latency of multi-device fine-tuning RoBERTa with the different bandwidth budget on the task of (a) CoLA and (b) MRPC, respectively.}
    \vspace{-5mm}
\label{fig: round_t_B}
\end{figure*}

Firstly, we evaluate the significance of backward propagation lengths on memory cost and computation latency.
The related costs and latency of fine-tuning a RoBERTa model toward the CoLA task with a batch size of 32 in a server are depicted in Fig.~\ref{fig: memory latency depth}.
Obviously, the memory load of optimizing the tunable parameters resided in different blocks is linearly increased with the block depth.
The increased memory consumption results from the fact that more intermediate results during the forward path are recorded to compute the gradient of parameters in deeper blocks.
Meanwhile, the average computation latency of one-step optimization is also semi-linearly increased with the block depth, which echoes the longer backpropagation time needed for computing more blocks sequentially.
It is also noted that the computation time for optimizing the deepest block is nearly twice that of optimizing the shallowest one, which causes significant computation diversity and urges the integrated C$^2$ consideration for multi-device cooperative fine-tuning. 

Next, we calculate the memory needed for fine-tuning a RoBERTa in a single device and in a multi-device-server cooperative manner.
For single-device fine-tuning, it not only fetches the model parameters but also records 1). all intermediate results,  2). all tunable parameters' gradients, and 3). statistics of gradients (mean and variance).
In multi-device DEFT, one device only needs to record the gradient of the assigned parameters and the associated intermediate results. The server is responsible for calculating the gradient statistics.
The memory cost per device is significantly reduced, as shown in Fig.~\ref{fig: deft memory}.
A single-device fine-tuning can consume nearly 12GB of memory, which is unaffordable for most current mobile devices.
Multi-device and server cooperation can substantially reduce on-device memory consumption, e.g., only $1/3$ of the single-device case, making the on-device running of fine-tuning algorithms more feasible.

\subsection{Performance of Block Allocation for LoLa DEFT}

The average round latency for multi-device DEFT considering varied transmit SNR and uplink bandwidth budget are depicted in Fig.~\ref{fig: round_t_snr} and Fig~\ref{fig: round_t_B} respectively.
The round latency is generally decreased with the increased transmit SNR or the enlarged bandwidth budget. 
It is observed that the depth-aware block allocation with averaged bandwidth assignment can dramatically reduce the one round C$^2$ latency compared with the baseline scheme, e.g., nearly 40\% reduction in the case of transmit SNR 10 dB. 
The joint bandwidth allocation scheme, i.e., JBBA, can further enhance latency reduction, which tends to align the C$^2$ latency of all involved devices by controlling everyone's communication latency.

\begin{figure*}[t]
    \centering
    \subfigure[CoLA]{\includegraphics[width=0.35\textwidth]{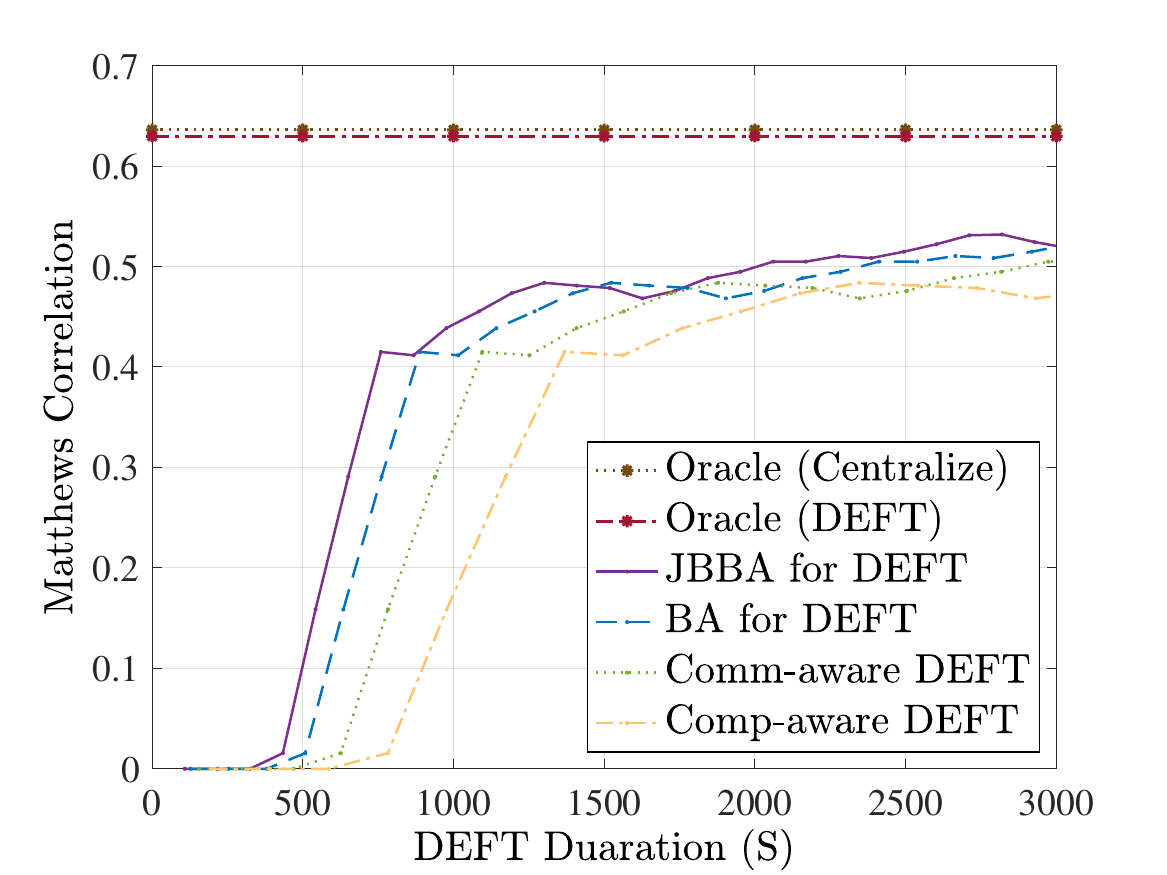}\label{fig: acc_t_cola}}
    \hspace{2.5cm}
    \subfigure[MRPC]{\includegraphics[width=0.35\textwidth]{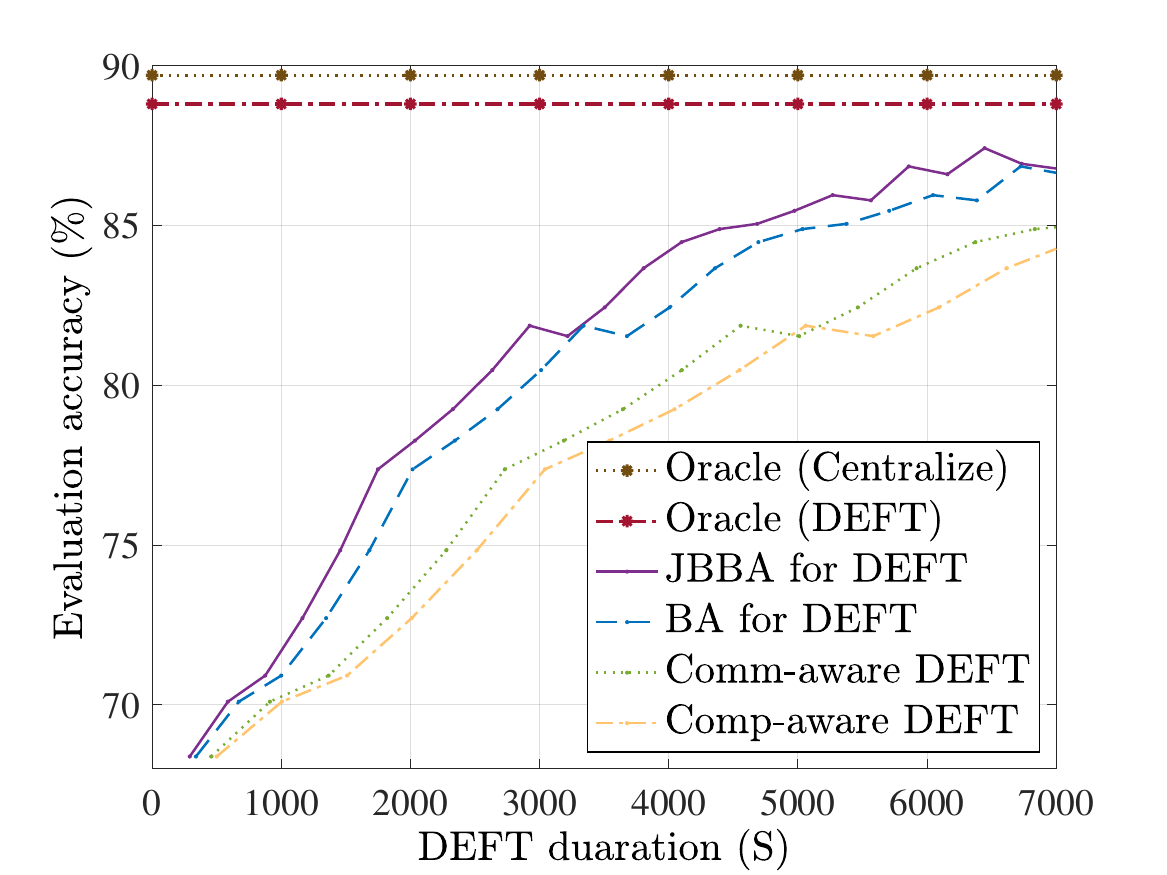}\label{fig: acc_t_mrpc}}
    \vspace{-2mm}
    \caption{The performance of fine-tuned RoBERTa on the task of (a) CoLA and (b) MRPC, respectively.}
\label{fig: acc_t}
\vspace{-2mm}
\end{figure*}

\begin{table*}[t]
\renewcommand\arraystretch{1.4}
    \caption{Duration of multi-device cooperative fine-tuning a RoBERTa base model toward tasks in GLUE adopting different block allocation schemes (in minutes).}
    \label{tab: DEFT durations}
    \centering
    \begin{small}
    \setlength{\tabcolsep}{2.8mm}{
    \begin{tabular}{|c|c|c|c|c|c|c|c|c|}
    \hline
    Schemes           & MNLI  & SST2  & MRPC & CoLA   & QNLI   & QQP   & RTE    & STSB  \\ \hline
    Comm-aware DEFT & 97.47 & 15.74 & 9.51 & 119.18 & 32.67  & 90.46 & 32.32  & 54.30 \\
    BA for DEFT    & 69.86 & 11.44 & 5.81 & 87.42  & 24.15  & 58.57 & 24.05  & 38.69 \\
    JBBA for DEFT   & 58.45 & 9.74  & 5.12 & 71.92  & 20.11  & 50.66 & 19.91  & 34.02 \\
    \hline
    \end{tabular}
    }
    \end{small}
\end{table*}

To further validate the effectiveness of depth-aware block allocation, we plot the performance curves along with the DEFT duration in Fig.~\ref{fig: acc_t}.
The oracle results for centralized fine-tuning are listed in \cite{LoRA}.
For DEFT, the oracle results are the highest correlation/accuracy over running the same number of optimization steps as the centralized one.
It is observed that the multi-device DEFT paradigm can achieve comparable performance with centralized tuning on both of the tasks considered.
Meanwhile, depth-aware block allocation can accelerate fine-tuning by minimizing the C$^2$ latency in every communication round.
Hence, the resultant depth-aware DEFT curve shifts the one-shot schedule curve toward the left in Fig~\ref{fig: acc_t}. 
Since the FoMo fine-tuning tends to saturate in the late training stage, the performance curve would vibrate, making it possible that the performance of BA for DEFT can surpass JBBA, at certain time points, e.g., 1200S and 2200S in the CoLA training curve.
In general, the performance curve of JBBA for DEFT encloses the others. 
We also list the time cost for multi-device DEFT to achieve a satisfactory performance using three schedule schemes under 8 GLUE tasks in Table~\ref{tab: DEFT durations}.  
The desired evaluation performance is set as $90\%$ of the centralized tuning case.
It is observed that BA and JBBA for DEFT schedules can reduce the time cost of the benchmark scheme by nearly $40\%$ and $30\%$ in all tasks, hence significantly accelerating the DEFT process. 

Finally, we investigate the effect of different numbers of devices for multi-device DEFT. The curves are plotted in Fig.~\ref{fig: round_t_K}. 
One can observe that the latency gaps are widened as the number of devices increases. 
In particular, the latency of JBBA is 40\% of the benchmark scheme and 60\% of the BA when 50 devices are participating in DEFT.
It is worth emphasizing that the advantage of JBBA, is more pronounced when the number of devices is large, which provides higher C$^2$ diversity for round-wise device scheduling.

\begin{figure*}[t]
    \centering
    \vspace{-2mm}
    \subfigure[CoLA]{\includegraphics[width=0.35\textwidth]{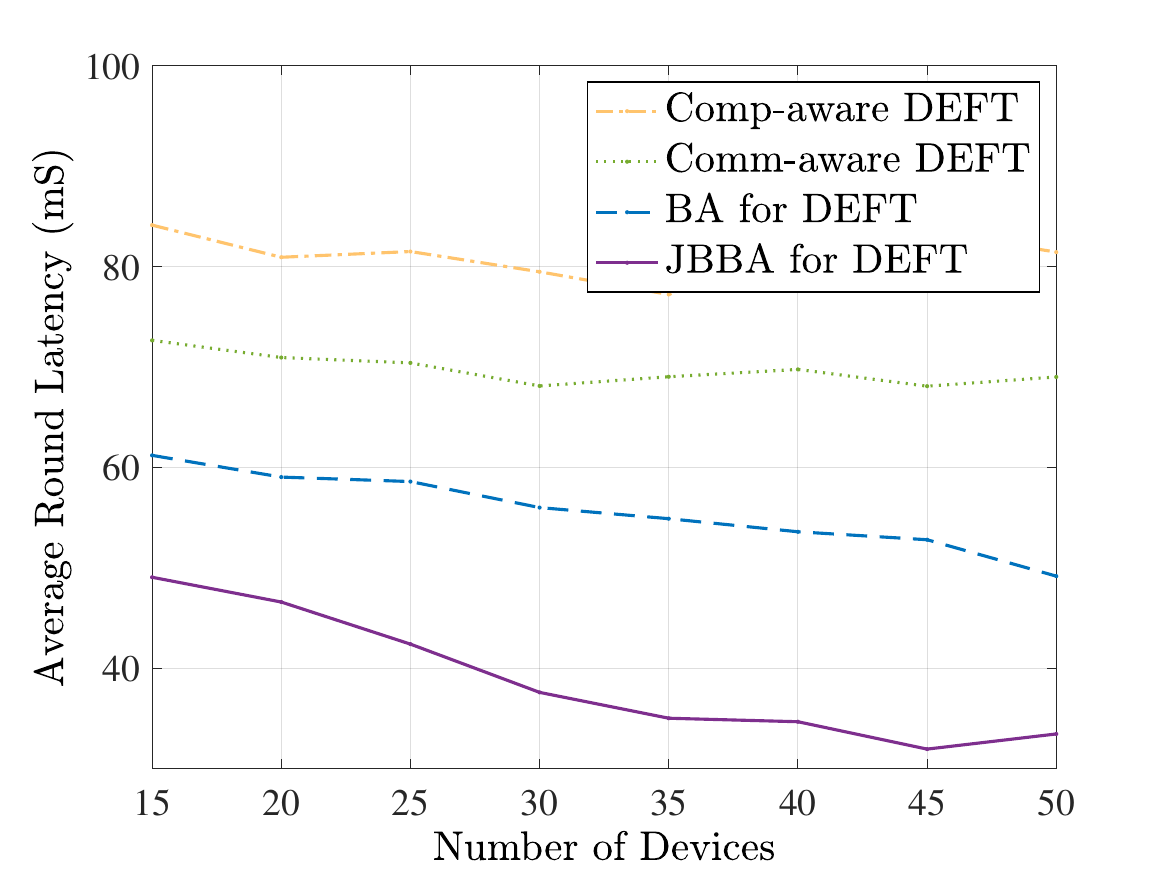}\label{fig: round_t_K_cola}}
    \hspace{2.5cm}
    \subfigure[MRPC]{\includegraphics[width=0.35\textwidth]{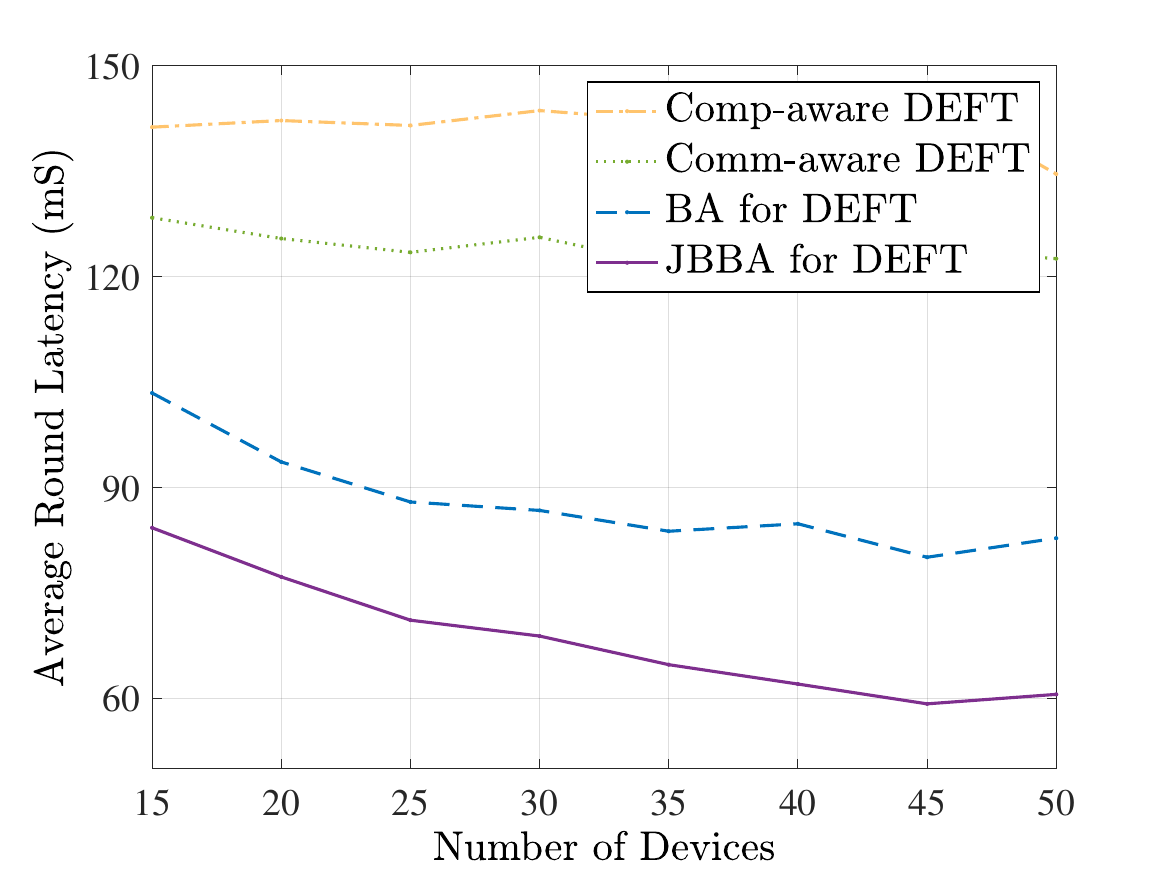}}\label{fig: round_t_K_mrpc}
    \vspace{-2mm}
    \caption{The averaged round latency of multi-device fine-tuning RoBERTa with the different number of devices on the task of (a) CoLA and (b) MRPC, respectively.}
\label{fig: round_t_K}
\vspace{-5mm}
\end{figure*}

\section{Concluding Remarks}
\label{sec: conclusions}

In this work, we consider the DEFT of foundation models at the wireless network edge and propose that multi-device cooperation can transcend the on-device memory limit memory and computation limits of single-device fine-tuning. 
Within the multi-device DEFT paradigm, it is found that the partitioned FoMo blocks residing in the different depths of a FoMo would result in varied computation latency and memory load.
Based on the block heterogeneity, the C$^2$ latency minimization for providing the possibility of fine-tuning large-scale models at the network edge is explored.
Minimizing C$^2$ latency is formulated into depth-aware block allocation problems associated with different bandwidth assignment cases, which are solved optimally with the proposed low-complexity CRUNCH algorithms, respectively. 
By depth-aware block allocation, multi-device DEFT can be accelerated dramatically, unleashing the potential of personalizing more state-of-the-art AI models for intelligent applications.

The exploitation of fine-tuning FoMo's for customizing devices' specifications opens a new direction for designing 6G task-oriented techniques to support distributed training and inference at the network edge.
The large-scale running of FoMo's at the wireless edge necessitates the consideration of energy efficiency, data swap, and user experience.
Specifically, DEFT involves the cooperation of devices and servers, where the coordination of computation offloading, training strategy, and communication setups is necessary.
Moreover, the interaction among devices and servers holding pre-trained FoMo's is worth investigating to employ collective intelligence for better utilization of FoMo's.

% \appendix

\bibliographystyle{IEEEtran}
\bibliography{Ref}

\end{document}